\documentclass[11pt,twocolumn]{IEEEtran}
\usepackage[latin1]{inputenc}
\usepackage{times,graphics,epsfig}
\usepackage{amssymb}
\usepackage{amsmath}
\usepackage{amsfonts}
\usepackage{amsfonts,amsmath,amstext}
\usepackage{psfrag,latexsym,verbatim,graphicx}
\usepackage{pstcol}
\usepackage{latexsym}
\usepackage{colordvi}
\usepackage{epsfig}
\usepackage{graphicx}
\newenvironment{myproof}[1]{\noindent\hspace{2em}{\itshape Proof #1:}}{\hspace*{\fill}~\QED\par\endtrivlist\unskip}
\newcommand{\IAN}{\mathrm {IAN}}
\newcommand{\JD}{\mathrm {JD}}
\newcommand{\OPT}{\mathrm {OPT}}
\newcommand{\SIC}{\mathrm {SIC}}
\newcommand{\SSS}{{P}}

\newcommand{\JJJ}{{I}}

\newcommand{\III}{{I}}

\newcommand{\stsets}[1]{\mathbb{#1}}
\newcommand{\RR}{\stsets{R}}

\newcommand{\EE}{\stsets{E}}
\newtheorem{theorem}{Theorem}

\newtheorem{lemma}{Lemma}
\newcommand{\singlespacing}{\let\CS=\@currsize\renewcommand{\baselinestretch}{0.95}\tiny\CS}
\newcommand{\oneandahalfspacing}{\let\CS=\@currsize\renewcommand{\baselinestretch}{1.25}\tiny\CS}
\newcommand{\doublespacing}{\let\CS=\@currsize\renewcommand{\baselinestretch}{1.39}\tiny\CS}
\newcommand{\be}{\begin{equation}}
\newcommand{\ee}{\end{equation}}

\newcommand{\C}{\mathcal{C}}
\newcommand{\R}{\mathcal{R}}
\newcommand{\CN}{\mathcal{CN}}
\newcommand{\bc}{\begin{center}}
\newcommand{\ec}{\end{center}}
\newcommand{\bfl}{\begin{flushleft}}
\newcommand{\efl}{\end{flushleft}}

\newcommand{\beqa}{\begin{eqnarray}}
\newcommand{\eeqa}{\end{eqnarray}}
\newcommand{\beqan}{\begin{eqnarray*}}
\newcommand{\eeqan}{\end{eqnarray*}}
\newcommand{\beq}{\begin{equation}}
\newcommand{\eeq}{\end{equation}}

\renewcommand{\S}{{\cal S}}

\newcommand{\eps}{\epsilon}

\newtheorem{thm}{Theorem}

\newtheorem{fact}[thm]{Fact}
\renewcommand{\S}{{\cal S}}

\newcommand{\T}{{\cal T}}
\newcommand{\V}{{\cal V}}
\newcommand{\Rv}{{\bf R}}
\newcommand{\Xv}{{\bf X}}

\newcommand{\B}{{\cal B}}
\newcommand{\A}{{\cal A}}

\newcommand{\D}{{\cal D}}
\newcommand{\sym}{\mathrm{sym}}

\newcommand{\mh}{\hat m}
\newcommand{\Mh}{\hat M}

\newcommand{\E}{{\cal F}}
\newcommand{\lbk}{\underline{k}}
\newcommand{\bI}{{\bar I}}

\title{Interference Networks with Point-to-Point Codes}
\author{Francois Baccelli, Abbas El Gamal, and David Tse\thanks{F. Baccelli was a Miller Professor at UC
Berkeley when this research was initiated. 
His work is supported in part by a grant of the INRIA@SiliconValley programme.
The research of A. El Gamal is supported in part by DARPA ITMANET.
The research of D. Tse is supported in part by the National Science Foundation under grant 0722032 and by
the AFOSR under grant number FA9550-09-1-0317.}\\
INRIA-ENS, Stanford University, UC Berkeley}

\date{Frebruary 1, 2011}

\begin{document}
\maketitle
\begin{abstract}
The paper establishes the capacity region of the Gaussian interference channel with many transmitter-receiver pairs constrained to use  point-to-point codes. The capacity region is shown to be strictly larger in general than the achievable rate regions when treating interference as noise, using successive interference cancellation decoding, and
using joint decoding. The gains in coverage and achievable rate using the optimal decoder are analyzed in terms of ensemble averages using stochastic geometry. In a spatial network where the nodes are distributed according to a Poisson point process and the channel path loss exponent is $\beta > 2$, it is shown that the density of users that can be supported by treating interference as noise  can scale no faster than $B^{2/\beta}$ as the bandwidth $B$ grows, while the density of users can scale linearly with $B$ under optimal decoding.
\end{abstract}
\begin{keywords}
Network information theory, interference, 
successive interference cancelation, joint decoding, stochastic geometry, coverage, ad hoc network,
stochastic network, performance evaluation.
\end{keywords}
\section{Introduction}
Most wireless communication systems employ point-to-point codes with receivers that treat interference as noise (IAN). This architecture is also assumed in most wireless networking studies. While using point-to-point codes has several advantages, including leveraging many years of development of good codes and receiver design for the point-to-point AWGN channel and requiring no significant coordination between the transmitters, treating interference as noise is not necessarily the optimal decoding rule. Motivated by results in network information theory, recent wireless networking studies have considered point-to-point codes with successive interference cancellation decoding (SIC) (e.g., see~\cite{WeberSIC}), where each receiver decodes and cancels the interfering codewords from other transmitters one at a time before decoding the codeword from its tagged transmitter, and joint decoding~\cite{BJ09} (JD), where the receiver treats the network as a multiple access channel and decodes all the messages jointly.

In this paper, we ask a more fundamental question: given that transmitters use point-to-point codes, what is the performance achievable by the optimal decoding rule?
The context we consider is a wireless network of multiple transmitter-receiver pairs, modeled as a Gaussian interference channel. The first result we establish in this direction is the capacity region of this channel when all the transmitters use Gaussian point-to-point codes. We show that none of the above decoding rules alone is optimal. Rather, a combination of treating interference as noise and joint decoding is shown to be capacity-achieving. Second, we show that this result can be extended to the case when the transmitters are only constrained to use codes that are capacity-achieving for the point-to-point and multiple access channels, but not necessarily Gaussian-like.

We then specialize the results to find a simple formula for computing the symmetric capacity for these codes. Assuming a wireless network model with users distributed according to a spatial Poisson process, we use simulations to study the gain in achievable symmetric rate and coverage when the receivers use the optimal decoding rule (OPT) for point-to-point Gaussian codes as compared to treating interference as noise, successive cancellation decoding, and joint decoding. We then use stochastic geometry techniques to study the performance in the wideband limit, where a high density of users share a very wide bandwidth. Under a channel model where the attenuation with distance is of the form $r^{-\beta}$ with $\beta > 2$, it is shown that the density of users that can be supported by treating interference as noise  can scale no faster than $B^{2/\beta}$ as the bandwidth $B$ grows, while the density of users can scale {\em linearly} with $B$ under optimal decoding. For an attenuation of the form $(k+r)^{-\beta}$, the density of users scales linearly with $B$, but when the distance between the tagged transmitter and its receiver tends to infinity, the rate for OPT scales like the wideband capacity of a point-to-point Gaussian channel without interference.

\section{Capacity Region with Gaussian Point-to-point Codes}\label{Prelim}
\label{sec:it}
Consider a Gaussian interference channel with $K+1$ transmitter-receiver pairs, where each transmitter $j\in [0:K]$ wishes to send an independent message $M_j \in [1:2^{nR_j}]$ to its corresponding receiver $j$ at rate $R_j$ (in the unit of bits/s/Hz). The signal at receiver $j$ when the complex signals $\Xv= (X_0,X_1,\ldots,X_{K})$ are transmitted is
\[
Y_j = \sum_{l=0}^{K} g_{jl} X_l + Z_j\quad \text{for } j \in [0:K],
\]
where $g_{jl}$ are the complex channel gains and $Z_j \sim \CN(0,1)$  is a complex circularly symmetric Gaussian noise with an average power of $1$. We assume each transmitter is subject to the same power constraint $Q$ (in the unit of Watts/Hz). Define the received power from transmitter $l$ at receiver $j$ as $P_{jl} = |g_{jl}|^2 Q$. Without further constraints on the transmitters' codes, the capacity region of this channel is not known even for the two transmitter-receiver pair case (see~\cite{ElGamal-Kim} for known results on this problem). In this section we establish the capacity region using Gaussian generated point-to-point codes for an arbitrary number of transmitter-receiver pairs.

We define an $(n,2^{nR_0},\ldots,2^{nR_K})$ \emph{Gaussian point-to-point} (G-ptp) code~\footnote{By a code here we just mean the message set and the codebook.} to consist of a set of randomly and independently generated codewords  $x_j^n(m_j)=(x_{j_1},x_{j_2},\ldots,x_{j_n})(m_j)$, $m_j \in [1:2^{nR_j}]$, $j \in [0:K]$, each according to an i.i.d. $\CN(0,\sigma^2)$ sequence, for some $0 < \sigma^2 \le Q$. We assume each transmitter in the Gaussian interference channel uses such a code with each receiver $j \in [0:K]$ assigning an estimate $\mh_j(y_j^n) \in [1:2^{nR_j}]$ of message $m_j$ to each received sequence $y_j^n$. We define \emph{the probability of error for a G-ptp code} as
\[
p_n = \frac{1}{K+1} \sum_{j=0}^KP\{\Mh_j \ne M_j\}.
\]
We denote the average of this probability of error over G-ptp codes as $\bar p_n$. A rate tuple $\Rv=(R_0,R_1,\ldots,R_K)$ is said to be achievable via a sequence of $(n,2^{nR_0},\ldots,2^{nR_K})$ G-ptp codes if $\bar p_n \to 0$ as $n \to \infty$. The \emph{capacity region with G-ptp} is the closure of the set of achievable rate tuples $(R_0,R_1,\ldots,R_K)$.

\noindent{\em Remarks:}
\begin{enumerate}
\item Our definition of codes precludes the use of time sharing and power control (although in general one can use time sharing with ptp codes). The justification is that time sharing (or the special cases of time/frequency division) require additional coordination.
\item Note that if a rate tuple is achievable via a sequence of G-ptp codes then there exists a sequence of (deterministic) codes that achieves this rate tuple. We use the definition of achievability via the average probability of error over codes to simplify the proof of the converse. The results, however, can be shown to apply to sequences of G-ptp codes almost surely, and to an even more general class of (deterministic) codes in  Section~\ref{sec:mac_cap}.
\end{enumerate}

Let $\S$ be a nonempty subset of $[0:K]$ and $\S^c= [0:K] \setminus \S$ be its complement. Define $X_{\S}$ to be the vector of transmitted signals $X_l$ such that $l \in \S$, and define the sum $X_j (\S)= \sum_{l \in \S} g_{jl} X_l$. Similarly define $P_j(\S)= \sum_{l \in \S} P_{jl}$, $R_{\S}$, and $R(\S) = \sum_{l \in \S} R_l$.

Consider a Gaussian multiple access channel (MAC) with transmitters $X_{\S}$, receiver $Y_j$, where $j \in \S$, and additive Gaussian noise power $P_j (\S^c) + 1$. Recall that the capacity region $\A_j(\S)$ of this MAC is
\[
\left\{R_\S:\,
R(\T) \le C \left(\frac{P_j (\T)}{1 + P_j (\S^c)} \right) \text{ for every } \T \subseteq \S\right\},
\]
where $C(x) = \log (1 +x)$ for $x \ge 0$. All logarithms are base $2$ in this paper.

Now, define the rate regions
\[
\C_j = \{\Rv:\, R_{\S} \in \A_j(\S) \text{ for some } \S \text{ containing $j$} \}.
\]
and
\begin{equation}
\label{eq:cap_reg}
\C = \bigcap_{j=0}^{K} \C_j.
\end{equation}
One of the main results in this paper is establishing the capacity region of the Gaussian interference channel with G-ptp codes.

\begin{theorem}\label{thm:capacity}
The capacity region of the  Gaussian $K+1$ transmitter-receiver pair interference channel with G-ptp codes is $\C$.
\end{theorem}

By symmetry of the capacity expression, we only need to establish achievability and the converse for the rate region $\C_0$, which ensures reliable decoding of transmitter $0$'s message at receiver $0$. Hence from this point onward, we focus on receiver $0$. We will refer to this receiver and its corresponding transmitter $0$ as \emph{tagged}. We also refer to other transmitters as {\em interferers}. We relabel the signal from the tagged receiver, its gains, and additive noise as
\[
Y = \sum_{l=0}^K g_l X_l  +Z.
\]
We also relabel the received power from the tagged  transmitter $0$ as $P_0$ and the received power from interferer $j$, $j\ge 1$, as $I_j$ (for interference). For any subset of interferers $\T$, we denote $I(\T)$ as the sum of the received power from these interferers. We will also drop the index $0$ from the notations $\A_0(\S)$ and $X_0(\S)$.

For clarity of presentation, first consider the case of $K=1$. Here the signal of the tagged receiver is
\begin{align*}
Y &= g_{0} X_0 + g_{1} X_1 + Z.
\end{align*}
For this receiver, there are two subsets to consider, $\S = \{0\}$ and $\S = \{0,1\}$. The region
$\A(\{0\})$ is the set of rate pairs $(R_0,R_1)$ such that
\[
R_0 \le C \left( \frac{P_0}{1 + I_1} \right),
\]
and the region $\A(\{0,1\})$ is the set of rate pairs $(R_0,R_1)$ such that
\begin{align*}
R_0 &\le C(P_0),\\
R_1 &\le C(I_1),\\
R_0+R_1 &\le C(P_0+I_1).
\end{align*}
Hence, the region $\C_0$ for the tagged receiver is the union of these two regions.

It is interesting to compare $\C_0$ to the achievable rate regions for other schemes that use point-to-point codes. Define the rate regions:
\begin{align*}
\R_{\rm IAN} & =  \A(\{0\}),\\
\R_{\rm SIC} & =  \left\{(R_0,R_1): R_0 \le C(P_0), 
R_1 \le C\Big(\frac{I_1}{1+P_0}\Big) \right\},\\
\R_{\rm JD} & =  \A(\{0,1\})~.
\end{align*}
The region $\R_{\rm IAN}$ is achieved by a receiver that decodes the tagged transmitter's message while treating  interference as Gaussian noise. The region $\R_{\rm SIC}$ is achieved by the successive interference cancellation receiver; the interferer's message is first decoded, treating the tagged transmitter's signal as Gaussian noise with power $\SSS$, and then the message from the tagged receiver is decoded after canceling the interferer's signal. The region $\R_{\rm JD}$ is the two transmitter-receiver pair Gaussian MAC capacity region. It is the set of achievable rates when the receiver insists on correctly decoding both messages, which is the achievable region using joint decoding in Blomer and Jindal \cite{BJ09}.

It is not difficult to see that the following relationships between the regions hold (see Figure~\ref{fig:cap_reg}):
\begin{align*}
\R_{\rm IAN} & \subset \C_0,\\
\R_{\rm SIC} &\subset \R_{\rm JD} \subset \C_0,\\
\C_0 &= \R_{\rm IAN} \cup \R_{\rm JD}.
\end{align*}
\begin{figure}[h]
\begin{center}
\small
\psfrag{r1}[b]{$R_1$}
\psfrag{r2}[l]{$R_0$}
\psfrag{a}[c]{$C_1\;\;\,$}
\psfrag{b}[c]{$C_{10}\;\;\;\,$}
\psfrag{c}[t]{$C_{01}$}
\psfrag{d}[t]{$C_0$}
\includegraphics[scale=0.45]{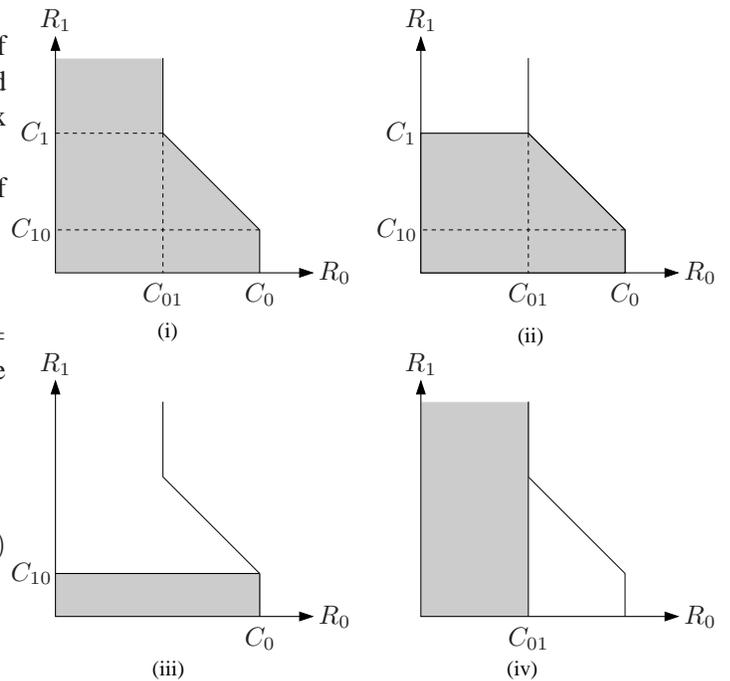}
\end{center}
\caption{$\C_0$ is the shaded region in Figure (i). The $\R_{\rm JD} $
region is depicted on Figure (ii).
$\R_{\rm SIC}$ is on Figure (iii) and $\R_{\rm IAN}$ is on (iv).
$C_0=C(P_0)$, $C_1=C(I_1)$, $C_{01}=C(P_0/(1+I_1))$, $C_{10}=C(I_1/(1+P_0))$.}
\label{fig:cap_reg}
\end{figure}
Note that the last relationship above says that the receiver can do no better than treating interference as Gaussian noise or jointly decoding the messages from the tagged transmitter and the interferer.

In the following, we first establish the capacity region for the case $K=1$, and then extend the result to arbitrary $K$. In Section \ref{sec:mac_cap}, we also show that our results extend to the class of \emph{MAC capacity-achieving codes}.
\subsection{Proof of Theorem~\ref{thm:capacity} for $K=1$ }

\noindent
{\em Proof of Achievability.} The prove the achievability of any rate pair in the interior of $\C_0$, we use Gaussian ptp codes with average power $Q(1- \delta)$ and joint typicality decoding as in~\cite{Cover--Thomas}. Further, we use simultaneous decoding~\cite{ElGamal-Kim} in which receiver $0$ declares that the message $\mh_0$ is sent if it is the unique message such that $(x_0^n(\mh_0),y^n)$ is jointly typical or $(x_0^n(\mh_0),x_1^n(\mh_1),y^n)$ is jointly typical \emph{for some} $m_1 \in [1:2^{nR_1}]$. A straightforward analysis of the average probability of error shows that $\bar p \to 0$ as $n \to \infty$ if either
\begin{equation}
\label{eq:ian}
R_0 < C\left( \frac{P_0 }{1 + I_1}\right),
\end{equation}
or
\begin{align*}
R_0 &< C(P_0),\\
R_0+R_1 &< C(P_0+I_1).
\end{align*}
The first constraint (\ref{eq:ian}) is $\A(\{0\})$, the IAN region. Denote the region defined by the second set of constraints by $\E(\{0,1\})$; it is the same as the MAC region $\A(\{0,1\})$ but with the constraint on $R_1$ removed. Hence, the resulting achievable rate region $\bar{C}_0 = \A(\{0\}) \cup \E(\{0,1\})$ appears to be larger than $\C_0= \A(\{0\}) \cup \A(\{0,1\})$. It is easy to see from Figure~\ref{fig:cap_reg}, however, that it actually coincides with $\C_0$. Hence, receiver $0$ can correctly decode $M_1$ if treating interference as noise fails but simultaneous decoding succeeds even though it does not require it. We establish the converse for the original characterization of $\C_0$, hence providing an alternative proof that the two regions coincide.

\noindent{\em Remark:} Although we presented the decoding rule as a two-step procedure, since the receiver knows the transmission rates, it already knows whether to apply IAN or simultaneous decoding.

\noindent
{\em Proof of the converse.} To prove the converse, suppose we are given a sequence of random G-ptp codes and decoders with rate pair $(R_0,R_1)$ and such that the average probability of error approaches $0$ as $n \to \infty$.  We want to show that $(R_0,R_1) \in \C_0$. Consider two cases:
\begin{enumerate}
\item $R_1 < C(I_1)$: Under this condition and by the assumption that the tagged receiver can reliably decode its message, the tagged receiver can cancel off the received signal from the tagged transmitter and then reliably decode the message from transmitter 1. Hence $(R_0,R_1)$ is in the capacity region of the MAC with transmitters $(X_0,X_1)$ and receiver $Y$, and hence in $\C_0$.

\item
$R_1 \ge C(I_1)$: Fix an $\eps >0$, and let $Z = U + V$, where $U$ and $V$ are independent Gaussian noise components with variances $N$ and $1-N$, respectively, such that
\[
C\Big(\frac{I_1}{N}\Big) = R_1 + \eps.
\]
Consider the AWGN channel
\beq
W =g_1 X_1 + U.
\eeq
Since we are assuming G-ptp codes and $R_1 < C(I_1/N)$, the average probability of decoding error over this channel approaches zero  as $n \to \infty$. Hence, by Fano's inequality, the mutual information over a block of $n$ symbols, averaged over G-ptp codes, is
\begin{eqnarray*}
\bar I(X_1^n; W^n ) & = &\bar h(X_1^n) -\bar h(X_1^n\mid W^n)\\
& \ge &  nR_1 - n \delta_n,
\end{eqnarray*}
where $\delta_n \to 0$ as $n\to \infty$. Denoting by $\bar h (W^n)$ the differential entropy of $W^n$ averaged over the G-ptp codes, this implies that
\begin{align*}
\bar h (W^n) & \ge   n R_1 - n\delta_n + h(U^n)\\
& =   nR_1- n \delta_n  + n \log( \pi e N)\\
& = n C\Big(\frac{I_1}{N}\Big)- n \eps - n \delta_n + n \log(\pi e N)\\
& =  n  \log (\pi e (I_1 + N)) - n \eps -n \delta_n.
\end{align*}
Now, let $\tilde{W}^n = W^n + V^n$.
By the conditional entropy power inequality, we have
\begin{align*}
2^{\frac{1}{n} \bar h(\tilde{W}^n)} & \ge   2^{\frac{1}{n} \bar h(W^n)} + 2^{\frac{1}{n}h(V^n)}\\
& \ge   2^{\log(\pi e (I_1 + N)) -  \delta_n - \eps} +
\pi e (1-N)\\
& =  \pi e (I_1 + N)2^{ -  \delta_n - \eps} + \pi e (1-N).
\end{align*}
Hence,
\begin{align*}
\bar h(\tilde W^n) & \ge  n \log \left( \pi e (I_1 + N)2^{ -  \delta_n - \eps}+ \pi e (1- N)\right).
\end{align*}
The fact that $\bar h(Y^n)\le n \log( \pi e (P_0 + I_1 +1))$ and the last lower bound
give an upper bound on the average mutual information for the tagged transmitter-receiver pair
\begin{align*}
\bar I(X_0^n;Y^n)
& =  \bar h(Y^n) - \bar h(\tilde{W}^n)\\
& \hspace{-.5cm} \le  n \log( \pi e (P_0 + I_1 +1)) \\ &
\hspace{-.5cm} -n \log( \pi e (I_1 + N)2^{-  \delta_n - \eps} + \pi e (1-N)).
\end{align*}
Since this is true for all $\eps > 0$, we have
\begin{align*}
 \bar I(X_0^n;Y^n)
& \le n \log(\pi e (P_0 + I_1+1))  \\ &
- n \log \left(\pi e (I_1 + N)2^{ -  \delta_n} + \pi e (1-N) \right)\\
& \le  n C\Big(\frac{P_0}{1+I_1}\Big) + n \tilde{\delta}_n.
\end{align*}
Since we assume the tagged receiver can decode its intended message, $R_0 < C(P_0/(1+I_1))$, and hence $(R_0,R_1) \in \C_0$. This completes the proof of Theorem~\ref{thm:capacity} for $K=1$.
\end{enumerate}

\noindent{\em Remarks:}
\begin{enumerate}
\item What the above proof  showed is that if the message of transmitter $0$ is reliably decoded, then either: (1) the interferer 's message can be jointly decoded as well, in which case the rate vector $\Rv$ is in the 2-transmitter MAC capacity region, or (2) the interference  plus the background noise is close to i.i.d. Gaussian, in which case decoding transmitter $0$'s message treating transmitter $1$'s interference plus background noise as Gaussian is optimal.

\item   One may think that since the interferer uses a Gaussian random code, the interference must be Gaussian and hence the interference plus background noise must also be Gaussian. This thinking is misguided, however, since what is important to the communication problem are the statistics of the interference plus noise {\em conditional on a realization of the interferer's random code}.  Given a realization of the code, the interference is discrete, coming from a code, and hence it is not in general true that the interference plus noise is close to i.i.d. Gaussian. What we showed in the above converse is that this holds when the message from the interferer cannot be jointly decoded with the message from transmitter $0$.
\end{enumerate}

\subsection{Proof of Theorem~\ref{thm:capacity} for arbitrary $K$ }
Now, consider the general case with $K+1$ transmitter-receiver pairs.

\noindent
{\em Proof of achievability.} The proof is a straightforward generalization of the proof for $K=1$, and the condition for the probability of error to approach $0$ is that the rate vector $\Rv$ lies in the region:
\begin{equation}
\label{eq:sim}
\bar{\C}_0 : = \{ \Rv: R_\S \in \E(\S) \text{ for some subset $\S$ with $0\in\S$}\},
\end{equation}
where
\begin{align*}
\E(\S) & =   \left\{\Rv: R(\T \cup \{0\})  < C\left ( \frac{P_0 + I(\T)}{1 + I(\S^c)}\right)\right.\\
& \hspace{1cm}\left. \text{ for every } \T \subseteq \S \setminus \{0\} \right\}
\end{align*}
is the \emph{augmented} MAC region for the subset of transmitters $\S$ treating the transmitters in $\S^c$ as Gaussian noise.

As in the $K=1$ case, the region $\bar{\C}_0$ appears to be larger than $\C_0$. We again establish the converse for the original characterization of $\C_0$, hence showing that $\bar{\C}_0$  coincides with $\C_0$.

\noindent
{\em Proof of the converse.} The proof for the $K=1$ case identifies, for a given a rate vector, a {\em maximal} set of interferers whose messages  can be jointly decoded with the tagged transmitter's message. This set depends on the given rates of the interferer; if $R_1 < C(I_1)$, the set is $\{1\}$, otherwise it is $\emptyset$. The key to the proof is to show that whichever the case may be, the residual interference created by the transmitters whose messages are {\em not} decoded plus the background noise must be asymptotically i.i.d. Gaussian.  We generalize this proof to an arbitrary number of interferers. In this general setting, however, {\em explicitly} identifying a maximal set of interferers whose messages can be jointly decoded with the tagged transmitter's message is a combinatorially difficult task. Instead, we identify it {\em existentially}.

Suppose the transmission rate vector is $\Rv$ and the average probability of error for the tagged receiver approaches zero as $n \to \infty$. Consider the set of subsets of interferers
\[
\D = \{ \T:  0 \not \in \T,\, R_\T \in \A(\T) \}.
\]
Intuitively, these are all the subsets of interferers whose messages can be jointly decoded after decoding $M_0$ while treating the other transmitted signals as Gaussian noise. Let $\T^*$ be a {\em maximal} set in $\D$, i.e., there is no larger subset $\T \in \D$ that contains $\T^*$. Since the message $M_0$ is decodable by the assumption of the converse, the tagged receiver can cancel off the tagged transmitter's signal. Next, the messages of the interferers in $\T^*$ can be decoded, treating the interference from the remaining interferers plus the background noise as Gaussian. This is because by assumption $R_{\T^*} \in \A(\T^*)$ and all interferers are using G-ptp codes. After canceling off the signals from the interferers in $\T^*$, the tagged receiver is left with interferers in $(\T^* \cup \{0\})^c$. Since no further messages can be decoded treating the rest as Gaussian noise (by the maximality of $\T^*$), it follows that for any subset $\S \subset (\T^* \cup \{0\})^c$, $R_\S$ is not in the capacity region of the MAC with transmitters in $\S$ and Gaussian noise with power $I((\T^* \cup \{0\})^c \setminus \S) + 1$. Let
\[
W = X((\T^* \cup \{0\})^c) + Z.
\]

In the $K=1$ scenario, $\T^*$ is either $\{1\}$ or $\emptyset$. In the first case, both messages are decoded, hence the power of the residual interference plus that of the background noise is automatically Gaussian. In the second case, the interferer's message is not decoded, and our earlier argument shows the interferer must be communicating above the capacity of the point-to-point Gaussian channel to receiver $0$. Hence the aggregate interference plus noise must be asymptotically i.i.d. Gaussian. In the general scenario with $K$ interferers, there may be more than one residual interferer left after decoding a maximal set $\T^*$. The following lemma, which is proved in the following subsection, shows that this situation generalizes appropriately.
\begin{lemma}
\label{lem:Gaussian}

Consider a $k$-transmitter MAC
\[
Y =\sum_{j=1}^k g_jX_j + Z,
\]
where the received power from transmitter $j$ is $P_j$ and $Z \sim \CN(0,1)$. Let
\begin{equation}
\label{eq:B}
\B = \{\Rv: R_\S \in \A(\S) \text{ for some nonempty } \S\}.
\end{equation}
If the transmitters use G-ptp codes at rate vector $\Rv$ and $\Rv \notin \B$, then
\[
\lim_{n \to \infty} \frac{1}{n} \bar h(Y^n ) = \log \left(\pi e \left (\sum_{j=1}^k P_j+1\right)\right),
\]
that is, the received sequence $Y^n$ is asymptotically i.i.d. Gaussian.
\end{lemma}

Lemma \ref{lem:Gaussian}  shows that the interference after decoding the interferers in $\T^* \cup \{0\}$ plus the background noise is asymptotically i.i.d. Gaussian. Hence, $R_{\T^* \cup \{0\}} \in \A(\T^* \cup \{0\})$, and we can conclude that $\Rv \in \C_0$. This completes the converse proof of Theorem \ref{thm:capacity} for arbitrary $K$.

\subsection{Proof of Lemma~\ref{lem:Gaussian}}
 The proof needs the following fact about $\B$. Recall that the boundary of the MAC capacity region consists of multiple faces. We refer to the one corresponding to the constraint on the total sum rate as the {\em sum rate face}.

\begin{fact}
\label{boundary}
Let $\Rv$ be a rate vector such that  $R_\S$ is on the boundary of $\A(\S)$ for some $\S$ but not on its sum-rate face. Then $\Rv$ cannot be on the boundary of $\B$. In other words, the non-sum-rate faces of the MAC regions $\A(\S)$ are never exposed on the boundary of $\B$.
\end{fact}

Figure \ref{fig:faces} depicts $\B$ for $K=2$. Here, the boundary of $\B$ consists of three segments, each of which is a sum-rate face of a MAC region. The two non-sum-rate faces of $\A(\{1,2\})$ are not exposed.

\begin{figure*}[tbp]
\begin{center}
\psfrag{r1}[b]{$R_2$}
\psfrag{r2}[l]{$R_1$}
\psfrag{a}[r]{$C_2$}
\psfrag{d}[t]{$C_1$}
\psfrag{r}[l]{$(R_1,R_2)$}
\psfrag{r3}[l]{$(R_1,R_2)$}
\psfrag{r4}[b]{$(R_1,R_2)$}
\includegraphics[scale=0.5]{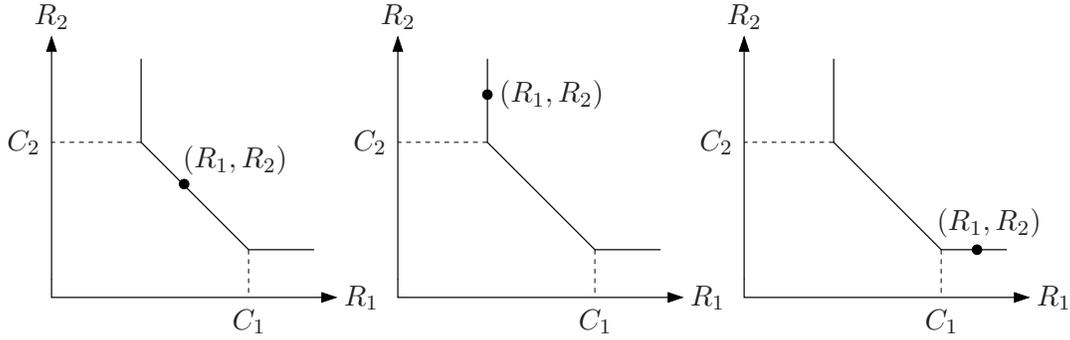}
\end{center}
\caption{The boundary of $\B$ for $K=2$ has three segments, all of which are sum-rate faces. A rate-tuple on the boundary of $\B$ can lie on one of them.}
\label{fig:faces}
\end{figure*}

\begin{myproof}{of Fact 1} Let $\Rv$ be a rate vector such that $R_\S$ is on the boundary of $\A(\S)$ for some $\S$ but not on its sum rate face. Then there is a  subset  $\T$ of $\S$ such that
\begin{equation}
\label{eq:tight}
R(\T) = C\left(\frac{P(\T)}{1+ P(\S^c)}\right)
\end{equation}
and for all subsets $\V$ strictly containing $\T$ and inside $\S$,
\beq
\label{eq:loose}
R(\V) < C\left(\frac{P(\V)}{1+ P(\S^c)}\right).
\eeq
Subtracting (\ref{eq:tight}) from (\ref{eq:loose}) implies that for all such sets $\V$,
\[
R(\V \setminus \T) < C\left(\frac{P(\V\setminus \T)}{1 + P(\T) + P(\S^c)}\right).
\]
This implies that $R_{\S\setminus \T}$ is in the strict interior of $\A(\S \setminus \T)$, Hence, $\Rv$ cannot be on the boundary of $\B$. This completes the proof of Fact 1.
\end{myproof}

\begin{myproof}{of Lemma~\ref{lem:Gaussian}} The proof is by induction on the number of transmitters $k$.

$k=1$: this just says that for a point-to-point Gaussian channel, if we transmit at a rate above capacity using a G-ptp code, then the output is Gaussian. This is a well-known fact.

Assume the lemma holds for all $j < k$. Consider the case with $k$ transmitters.

Express $Z = U+V$, where $U$ and $V$ are independent Gaussians with variances $N$ and $1-N$, respectively, where $N$ is chosen such that $\Rv$ is on the boundary of $\tilde{\B}$ for the MAC
\[
\tilde W = \sum_{j=1}^k g_j X_j + U.
\]
Here, $\tilde{\B}$ is the same as $\B$ except that the background noise power $1$ is replaced by $N$. Let ${\cal E}$ be the collection of all subsets $\S \subset [1:k]$ for which $R_\S \in \tilde{\A}(\S)$ ($\tilde{\A}(\S)$ is the same as $\A(\S)$ except that the background noise power $1$ is replaced by $N$). Pick a maximal subset $\S$ from that collection.
 By Fact~\ref{boundary}, $R_\S$ must be on the sum-rate face of $\tilde{\A}(\S)$. The MAC can be decomposed as
\[
\tilde W = X(\S) + X(\S^c) + U.
\]
By the maximality of $\S$, no further transmitted messages can be decoded beyond the ones for the transmitters in $\S$ (otherwise, there would exist a bigger subset $\S'$ containing $\S$ and for which $R_{\S'} \in \tilde{\A}(\S')$). This implies in particular that for any subset $\T \subset S^c$, the rate vector $R_{\T}$ cannot be in the region $\tilde{\A}(\T)$; otherwise if such a $\T$ exists, the receiver could have first decoded the messages of transmitters in $\S$, cancelled their signals, and then decoded the messages of the transmitters in $\T$, treating the residual interference plus noise as Gaussian. Hence if we consider the smaller MAC
\[
 W = X(\S^c) + U,
 \]
we can apply the induction hypothesis to show that $W^n$ is asymptotically i.i.d. Gaussian. So now we have a Gaussian MAC for transmitters in $\S$
\[
\tilde W  = X(\S) + W
\]
and since the rate vector $R_\S$ lies on the sum rate boundary of this MAC, we now have a situation of a super-transmitter, i.e., a combination of all transmitters in $\S$,  sending at the capacity of this Gaussian channel. Using a very similar argument as in the $K=1$ proof, one can show that $W^n$ is asymptotically i.i.d. Gaussian. Adding back the removed noise $V$ yields the desired conclusion. This completes the proof of Lemma~\ref{lem:Gaussian}.
\end{myproof}

\section{Capacity Region with MAC-Capacity-Achieving Codes}\label{sec:mac_cap}

The converse in Theorem \ref{thm:capacity} says that if the transmitters use Gaussian random codes, then one can do no better than treating interference as Gaussian noise or joint decoding. The present section shows that this converse result generalizes to a certain class of (deterministic) ``MAC-capacity-achieving" codes, to be defined precisely below. We first focus on the two-transmitter-receiver pair case and then generalize to the $K+1$-transmitter case.

An $(n,2^{nR})$ (deterministic) single-user code satisfying the transmit power constraint $Q$ is said
to achieve a rate $R$ over a point-to-point Gaussian channel $Y=gX + Z$ if the probability of error $p_n \to 0$ as the block length $n \to \infty$. An $(n,2^{nR})$ code is said to be {\em point-to-point (ptp) capacity-achieving} if it achieves a rate of $R$ over {\em every} point-to-point Gaussian channel with capacity greater than $R$.

Now consider the two transmitter-receiver pair Gaussian interference channel.
 A rate-pair $(R_0,R_1)$ is said to be achievable over the interference channel via a sequence of ptp-capacity-achieving codes if there exists a sequence of such codes for each transmitter such that the probability of error
 \[
p_n = \frac{1}{2} \left( P\{\hat{M}_0 \ne M_0\} + P\{\hat{M}_1 \ne M_1\}\right)
\]
approaches $0$ as $n \to \infty$.  The capacity region with ptp-capacity-achieving codes is the closure of the set of achievable rates. The theorem below is a counterpart to the converse in Theorem \ref{thm:capacity} for G-ptp codes.

\begin{theorem}
\label{thm:2-Tx}
The capacity region of the two transmitter-receiver pair interference channel with ptp-capacity achieving codes is no larger than $\C$,
as defined in (\ref{eq:cap_reg}) for $K=1$.
\end{theorem}

\begin{proof}
The result follows from the observation that in the proof of the converse for Theorem \ref{thm:capacity}, the only property we used about the G-ptp codes is that the average decoding error probability of the interferer's message after canceling the message of the intended transmitter goes to zero whenever $R_1 < C(I_1)$. This property remains true if the interferer uses a ptp-capacity-achieving code instead of a G-ptp code.

\end{proof}

Theorem \ref{thm:2-Tx} says that as long as the codes of the transmitters are designed to optimize point-to-point performance, the region $\C$ is the fundamental limit on their performance over the interference channel. This is true even if the codes do not ``look like" randomly generated Gaussian codes.

Now let us consider the $K+1$-transmitter interference channel for general $K$. Is $\C$ still an outer bound to the capacity region if all the transmitters use ptp-capacity-achieving codes? The answer is no. A counter-example can be found in \cite{BPT10} (Section IIB), which considers a 3-transmitter many-to-one interference channel with interference occurring only at receiver $0$. There, it is shown that if each of the transmitters uses a lattice code, which is ptp-capacity-achieving, one can do better than {\em both} joint decoding  all transmitters' messages and decoding just transmitter $0$'s message treating the rest of the signal as Gaussian noise at receiver $0$. The key is to use lattice codes for transmitter $1$ and $2$, and have them {\em align} at receiver $0$ so that the two interferers appear as one interferer. Hence, it is no longer necessary for receiver $0$ to decode the messages of {\em both} interferers in order to decode the message from transmitter $0$; decoding the {\em sum} of the two interferers is sufficient. At the same time, treating the interference from $1$ and $2$ as Gaussian noise is also strictly sub-optimal.

In this counter-example, the transmitters' codes are ptp-capacity-achieving but not "MAC capacity-achieving" in the sense that receiver $0$ cannot jointly decode the individual messages of the interferers. A careful examination of the proof of the converse in Theorem \ref{thm:capacity} for general $K$ reveals that the converse in fact holds whenever the codes of the transmitters satisfy such a  MAC-capacity-achieving property.

Consider a $k$-transmitter Gaussian MAC
\[
Y=\sum_{j=1}^k g_j X_j + Z
\]
and a subset $\S \subset [1:k]$. A $(n,2^{nR_1}, \ldots, 2^{nR_k})$ (deterministic) code for this MAC, where each transmitter satisfies the same transmit power constraint $Q$, is said
to achieve the rate-tuple $R_\S$ over the MAC if the probability of error
\[
p_n(\S) = \frac{1}{|\S|} \sum_{j\in \S} P\{\Mh_j \ne M_j\}
\]
approaches $0$ as $n \to \infty$. An $(n,2^{nR_1}, \ldots, 2^{nR_k})$ code is said to be {\em MAC-capacity-achieving} if for every $\S \subset [1:k]$, it achieves a rate $R_\S$ over {\em every} Gaussian MAC whose capacity region $\A(\S)$ contains $R_\S$. Recall that the region $\A(\S)$ is the capacity region of the MAC with transmitters $X_\S$  and the signals from the rest of the transmitters treated as Gaussian noise. Thus this definition says that a MAC capacity-achieving code is good enough to achieve this performance for any subset $\S$ of transmitters.

Now consider the $K+1$ transmitter-receiver pair Gaussian interference channel.
A rate-tuple $\Rv$ is said to be achievable on the interference channel via a sequence of MAC-capacity-achieving codes if there exists a sequence of MAC-capacity-achieving codes for every subset containing $K$ transmitters such  that the probability of error
\[
p_n = \frac{1}{K+1} \sum_{j=0}^K P\{\hat{M}_j \not = M_j\}
\]
approaches zero as $n \to \infty$.  The capacity region with MAC-capacity-achieving codes is the closure of all such rates.

\begin{theorem}
The capacity region of the Gaussian $K+1$-transmitter interference channel with MAC-capacity achieving codes is no larger than $\C$, as defined in (\ref{eq:cap_reg}).
\end{theorem}

\begin{proof}

The result follows from the observation that in the proof of the converse in Theorem \ref{thm:capacity}, the only property that was used about the G-ptp codes of the transmitters  is precisely the MAC-capacity-achieving property defined above.
\end{proof}

The counter-example above shows that one can indeed do better than the region $\C$, for example using interference alignment. Interference alignment, however, requires careful coordination and accurate channel knowledge at the transmitters. On the other hand, one can satisfy the MAC-capacity-achieving property without the need of such careful coordination. So, if one takes the MAC-capacity-achieving property as a definition of  lack of coordination between the transmitters, then the above theorem delineates the fundamental limit to the performance on the interference channel if the transmitters  are not coordinated.


\section{Symmetric Rate}\label{secSD}

We specialize the results in the previous sections to the case when all messages have the same rate $R$. This will help us compare the network performance of the optimal decoder to other decoders for Gaussian ptp codes.
Throughout the section, we assume that $I_1 \ge I_2 \cdots \ge I_K$, and define $I[j:k] = \sum_{i=j}^k I_i$ and $I = \sum_{i=1}^K I_i$. When $K = \infty$, we will assume that $I$ is finite, hence $I_i \to 0$ as  $i \to \infty$.

\subsection{Optimal Decoder}\label{secSDO}
Focusing again on the tagged receiver $0$, define the \emph{symmetric rate} $R_\sym$  as the supremum over $R$ such that $(R,R,\ldots,R) \in \C_0$.
We can express the symmetric rate $R_\sym$ as the solution of a simple optimization problem.

\begin{lemma}\label{lem:symmetric}
The symmetric rate
under G-ptp codes is
\begin{equation}
\label{eq:infinity}
R_\sym=
\max_{k \in [0:K]} \min_{l \in [0:k]} \frac{1}{l+1} C\left( \frac{P_0 + I[k-l+1:k]}{1 + I[k+1:K]}\right).
\end{equation}
\end{lemma}

\begin{proof}
From the reduced characterization of $\C_0$ in~\eqref{eq:sim}, we have
\[
R_\sym = \max_{\S: 0\in \S} R_\sym(\S) = \max_{k \in [0:K]} R_\sym([0:k]),
\]
where $R_\sym(\S)$ is the symmetric rate of the region $\E(\S)$. The second equality follows from the observation that the reduced MAC region $\E(\S)$ is monotonically increasing in the received powers from the transmitters in $\S$ and decreasing in the interference power from transmitters in $\S^c$. Hence, among all subsets $\S$ of size $k+1$, the one with the largest symmetric rate is $[0:k]$ (the one with the highest powered transmitters and lowest powered interferers).

Taking into account all $2^{k}$ constraints of the region $\E([0:k])$, we have
\begin{align*}
&  R_\sym([0:k]) \\
& = \min_{\T \subset [1:k]} \frac{1}{|\T|+1}C\left(\frac{P_0+I(\T)}{1+ I[k+1:K]}\right)\\
& = \min_{l \in [0:k]} \min_{\T \subset [1:k],|T| =l}
\frac{1}{l+1}C\left(\frac{P_0+ I(\T)}{1+ I[k+1:K]}\right).
\end{align*}
The desired result (\ref{eq:infinity}) now follows from the fact that among all the subsets $\T$ of size $l$, the one with the smallest total power $I(\T)$ is $[k-l+1:k]$.
\end{proof}

\subsection{Other Decoders}
We will use the following nomenclature  for the rest of the paper:
\begin{itemize}
\item IAN refers to treating interference as noise decoding. The condition for IAN is
\[
R < C\left (\frac{P_0}{ 1 + I} \right).
\]
\item SIC($k$) refers to successive interference cancellation in which the tagged receiver sequentially decodes and cancels  the signals from the $k$ strongest transmitters treating other signals as noise and then decodes the message from the tagged transmitter while treating the remaining signals as Gaussian noise. The conditions for SIC are
\begin{align*}
\label{eq:sic_cond}
R &< C\left (\frac{I_l}{1 + P_0 + I[l+1:K]} \right) \text{for } l \in [1:k],\\
R &< C\left (\frac{P_0}{1 + I[k+1:K]} \right).
\end{align*}
\item JD($k$) refers to joint decoding of the messages of the first $k+1$ transmitters and treating the rest as Gaussian noise. The conditions for JD($k$) are
\begin{equation}
\label{eq:JD_cond}
R < \frac{1}{l+1} C\left( \frac{P_0 + I[k-l+1:k]}{1 + I[k+1:K]}\right) \text{for } l \in [0:k].
\end{equation}
The JD($k$) conditions are not monotonic in $k$, that is, the fact that JD($k$) holds
neither implies that JD($k'$) holds for $k'<k$ nor for $k'>k$ in general.
\item OPT($k$) refers to the optimal decoder used in the proof of Theorem \ref{thm:capacity}
if there were only $k$ interferers. OPT($K$) or simply OPT refers to the optimal decoding rule.
The conditions for OPT($K$) are $R<R_\sym$ with
$R_\sym$ given by (\ref{eq:infinity}).
Since the condition for OPT is the union of the JD($l$) conditions for $0\le l\leq k$,
if OPT($k$) holds, so does OPT($k'$) for all $k'>k$.
\end{itemize}

\subsection{Number of Interferer Messages Decoded}

Lemma \ref{lem:symmetric} shows that, for $K$ finite, the optimal decoding strategy
is to use $\JD(k_{\mathrm {opt}})$ with
\begin{equation}
\label{eq:optimalk}
k_{\mathrm {opt}}= {\mathrm {argmax}_{k \in [0:K]}} \xi(k),
\end{equation}
where
\begin{equation}
\label{eq:xifunction}
\xi(k)=\min_{l \in [0:k]} \frac 1 {l+1}
C\left( \frac{P_0 + I[k-l+1:k]}{1 + I[k+1:K]}\right)
\end{equation}
provided the argmax in question is uniquely defined.

The following lemma is focused on the case $K=\infty$, which will be considered
in the next sections, and where one may fear that the maximum is not defined
in (\ref{eq:infinity}), i.e., the argmax in (\ref{eq:optimalk}) is not defined. Fortunately, this is not the case.

\begin{lemma}\label{lem:kstar}
If $K=\infty$, $P_0>0$ and $I<\infty$, then $k_{\mathrm {opt}}<\infty$.
\end{lemma}

\begin{proof}
We have
\begin{align*}
0\le \xi(k) & \le   \min_{l \in [0:k]}  \frac{1}{l+1} C(P_0+I) 
= \frac{1}{k+1} C(P_0+I).
\end{align*}
Hence $k\to \xi(k)$ is a positive function bounded from above by a function that tends to 0 when $k$ tends to infinity. The values where $\xi$ is maximal are then all finite unless it is 0 everywhere. But this is not the case since our assumptions on $P_0$ and $I$ imply that $\xi(0)>0$.
\end{proof}

The following lemma will be used later.

\begin{lemma}\label{cor1}
Let
\[
\lbk = \min \{ k \ge 1\text{ such that } I_k < P_0 \}.
\]
Then, a sufficient condition for achievability by OPT at rate $R$ is that
\begin{equation}
\label{eq:boundlp}
\lbk R < C\left(\frac{\lbk P_0}{1 + I[\lbk:\infty]} \right).
\end{equation}
Further, if this condition holds, then the conditions for JD($\lbk -1$) are met.
\end{lemma}

\begin{proof}
We can derive a lower bound for the symmetric rate in (\ref{eq:infinity}) in terms of $\lbk$:
\begin{align*}
& \max_{k \in [0:K]} \min_{l \in [0:k]} \frac{1}{l+1} C\left( \frac{P_0 + I[k-l+1:k]}{1 + I[k+1:K]}\right)\\
& \ge \min_{l \in [0:\lbk-1]} \frac{1}{l+1} C\left( \frac{P_0 + I[\lbk-l:\lbk-1]}{1 + I[\lbk:K]}\right)\\
& \ge \min_{l \in [0:\lbk-1]} \frac{1}{l+1} C\left( \frac{(l+1)P_0}{1 + I[\lbk:K]}\right)\\
& =  \frac{1}{\lbk} C\left( \frac{\lbk P_0}{1 + I[\lbk:K]}\right).
\end{align*}
The first inequality is obtained by choosing $k=\lbk-1$ in the outer maximization; the second inequality is obtained by lower bounding the received powers of all the interferers with index $\le \lbk$ by $P_0$; the last equality follows from the fact that $C(x)/x$ is a monotonically decreasing function of $x$. The sufficient condition (\ref{eq:boundlp}) for achievability is now obtained by requiring the target rate $R$ to be less than this lower bound.
\end{proof}

Lemma \ref{cor1} gives a guideline on how to select the set of interferers to jointly decode: under the condition (\ref{eq:boundlp}), the success of joint decoding at rate $R$ is guaranteed when decoding all interferers with a received power larger than that of the tagged transmitter.  This is only a bound, however, and as we will see in the simulation section, one can often succeed in decoding more than $\lbk-1$ transmitters.


\section{Spatial Network Models and Simulation Results}
The aim of the simulations we provide in this section is to
illustrate the performance improvements of OPT versus IAN and JD.
The framework chosen for these simulations is a spatial network with
a denumerable collection of randomly located nodes. In the following section we also use this spatial network model  for mathematical analysis.
\subsection{Spatial Network Models}\label{SpatNet}
All the spatial network models considered below feature a
set of transmitter nodes located in the Euclidean plane.
The channel gains defined in Section \ref{Prelim}, or equivalently
the received signal power $P_0$ and the interference powers $I_j$, $j \in [1:K]$, at the tagged receiver
are evaluated using a path loss function $l(r)$, where $r$ is distance.
Here are two examples used in the literature (and in some examples below):
\begin{itemize}
\item $l(r)=r^{-\beta}$, with $\beta>2$ (case with pole),
\item $l(r)=(k+r)^{-\beta}$, with $\beta>2$ and $k$ a constant (case without pole); it makes
sense to take $k$ equal to the wavelength.
\end{itemize}

More precisely, if we denote the locations of the transmitters by $T_j$, $j \in [0:K]$ and that of the tagged receiver by $y$
and if we assume that the tagged receiver selects the interferers with the strongest received powers to be jointly decoded, then
$|g_{00}|^2=l(|T_0-y|)$ and $|g_{0j}|^2=l(|T_j-y|)$, or equivalently
$P_0 = l(|T_0-y|) Q $ and $I_j= l(|T_j-y|) Q$ for $j \in [1:K]$. Here $Q$ denotes the transmit power.
Since we assume that $I_1 \ge I_2 \ldots \ge I_K$, the strongest interferer is the closest one to $y$ (excluding the tagged transmitter).
Let $I(y)$ be the total interference at the tagged receiver, namely $I(y) =\sum_{j\ne 0}  I_j $.

The simulations also consider the following extensions of this basic model:
\begin{itemize}
\item The fading case, where the channel gain is further multiplied by $F_j(y)$, where
$F_j(y)$ represents the effect of fading from transmitter $j$ to $y$. In this case, the strongest
interferer is not necessarily the closest to $y$.
\item The case where the power constraint is not the same for all transmitters.
Then $P_0 = l(|T_0-y|) Q_0 $ and $I_j= l(|T_j-y|) Q_j$ for $j \in [1:K]$, with $Q_j$
the power constraint of transmitter $j$.
\end{itemize}

\subsection{IAN, SIC($1$), JD($1$), and OPT($1$) Cells}\label{sec:cells}
\subsubsection{Definitions}
Fix some rate $R$.
For each decoding rule $A$ (i.e., IAN, \ldots) as defined above,
let $\Xi^{A}$ be
the set of locations in the plane where the conditions for rule $A$ are met with respect to the tagged transmitter
and for $R$. We refer to this set as the $A$ cell for rate $R$. The main objects of interest are hence the cells
$\Xi^{\IAN}$, $\Xi^{\SIC(1)}$, $\Xi^{\JD(1)}$ and $\Xi^{\OPT(1)}$.

\subsubsection{Inclusions}
Rather than looking at the increase of rate obtained when moving from
a decoding rule to another, we fix $R$ and
compare the cells of the two decoding rules.
In view of the comparison results in Section \ref{sec:it}, we have
\begin{align*}
& \Xi^{\IAN} \subset \Xi^{\OPT(1)},\\
& \Xi^{\SIC(1)} \subset \Xi^{\JD(1)} \subset \Xi^{\OPT(1)},\\
& \Xi^{\OPT(1)}= \Xi^{\IAN} \cup \Xi^{\JD(1)}.
\end{align*}

For all pairs of conditions $A$ and $B$, we define  $\Xi^{A\setminus B}$ to be the set of
locations in the plane where the condition for $A$ is met but the condition for $B$ is not met.
For instance,
\begin{align}
\Xi^{\SIC(1)\setminus \IAN} &=  \Xi^{\SIC(1)}\setminus \Xi^{\IAN}.\nonumber
\label{eq:cellincr}
\end{align}

\subsubsection{Simulations}
In the simulation plots below,
the transmitters are randomly located according to a Poisson point process. The attenuation function
is of the form $l(r)=(1+r)^{-\beta}$ or $l(r)=r^{-\beta}$.

Figure~\ref{fig:sim1} compares $\Xi^{\SIC(1)}$ and $\Xi^{\OPT(1)}$.
Notice that SIC does not increase the
region that is {\em covered} compared to IAN, whereas OPT(1) does.

\begin{figure}[h]
\begin{center}
\includegraphics[width=0.33\textwidth]{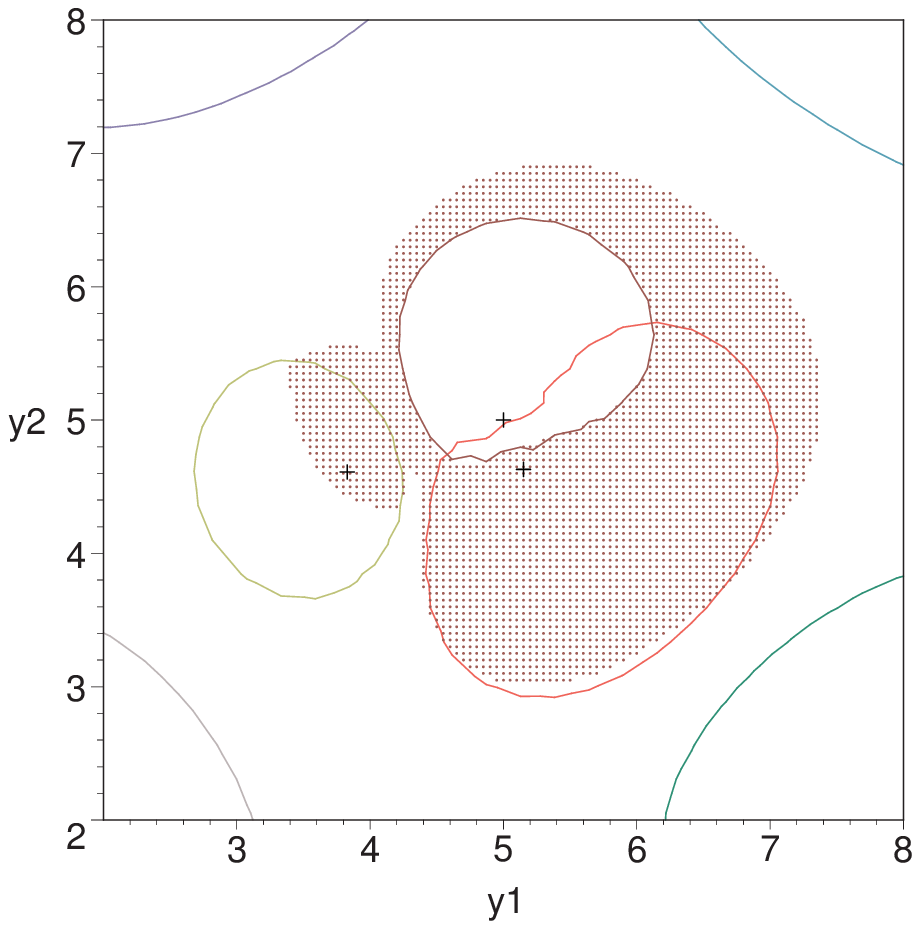}
\includegraphics[width=0.33\textwidth]{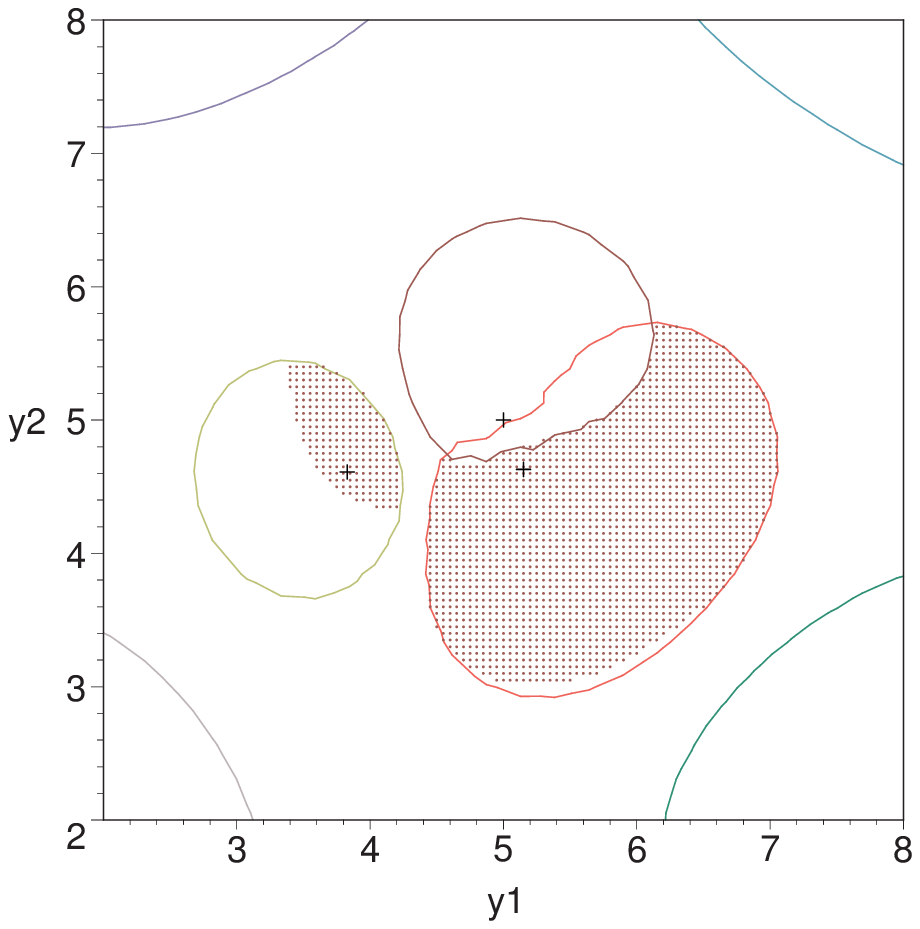}
\end{center}
\caption{The top plot is for
$\Xi^{\OPT(1)\setminus \IAN}$ and the bottom plot is for
$\Xi^{\SIC(1)\setminus \IAN}$ for the tagged transmitter. The transmitters are denoted by crosses.
The contours denote the boundaries of the IAN cells of different transmitters.
The spatial user density is 0.1. The power constraints $Q_i$ are here
randomly chosen according to a uniform  distribution over $[0,2000]$. Variable transmission powers
show up when devices are heterogeneous or power controlled. $R=0.73$ bits/s/Hz and $\beta=3$.
The tagged transmitter is at the center of the plot (at $[5,5]$). The attenuation is $l(r)=(1+r)^{-\beta}$.
}
\label{fig:sim1}
\end{figure}

Figure~\ref{fig:sim2} compares OPT($1$) to JD($1$) and IAN. Note that there is no gain moving from JD($1$)
to OPT($1$) outside the IAN cell.  Also, in such a spatial network, one of
the practical weaknesses of JD($1$) is its lack of coverage continuity (the
JD($1$) cell has holes and may even lack connectivity as shown in
the plots). These holes are due to the unnecessary symmetry between
the tagged transmitter and the strongest interferer, which
penalizes the former.

\begin{figure*}[tbp]
\begin{center}
\includegraphics[width=0.3\textwidth]{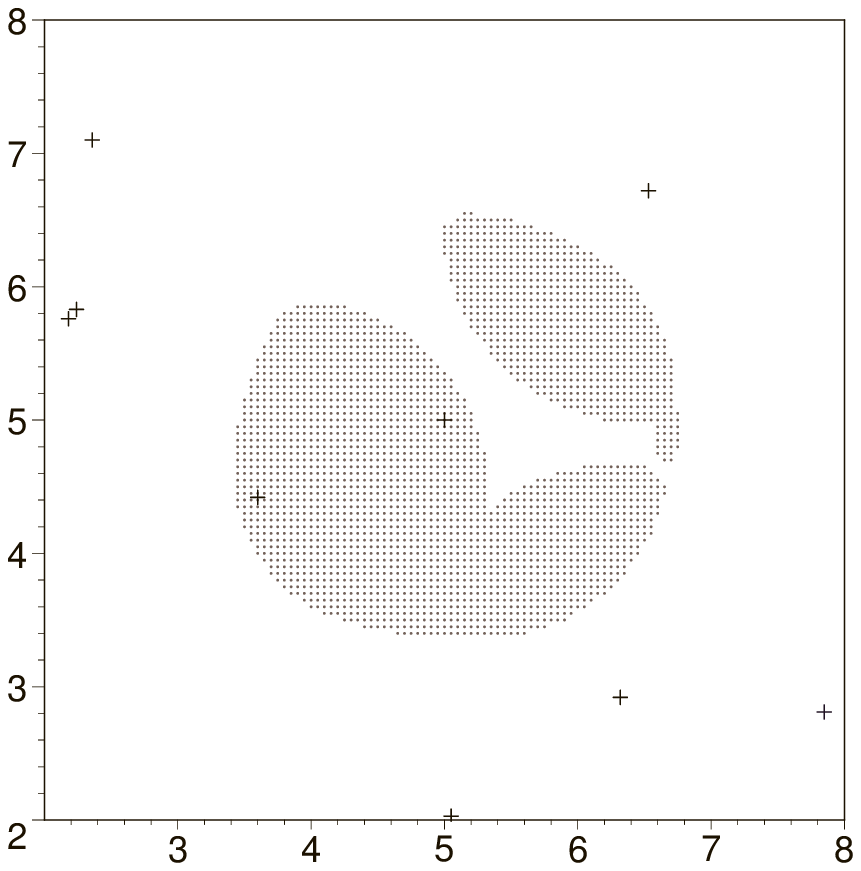}
\includegraphics[width=0.3\textwidth]{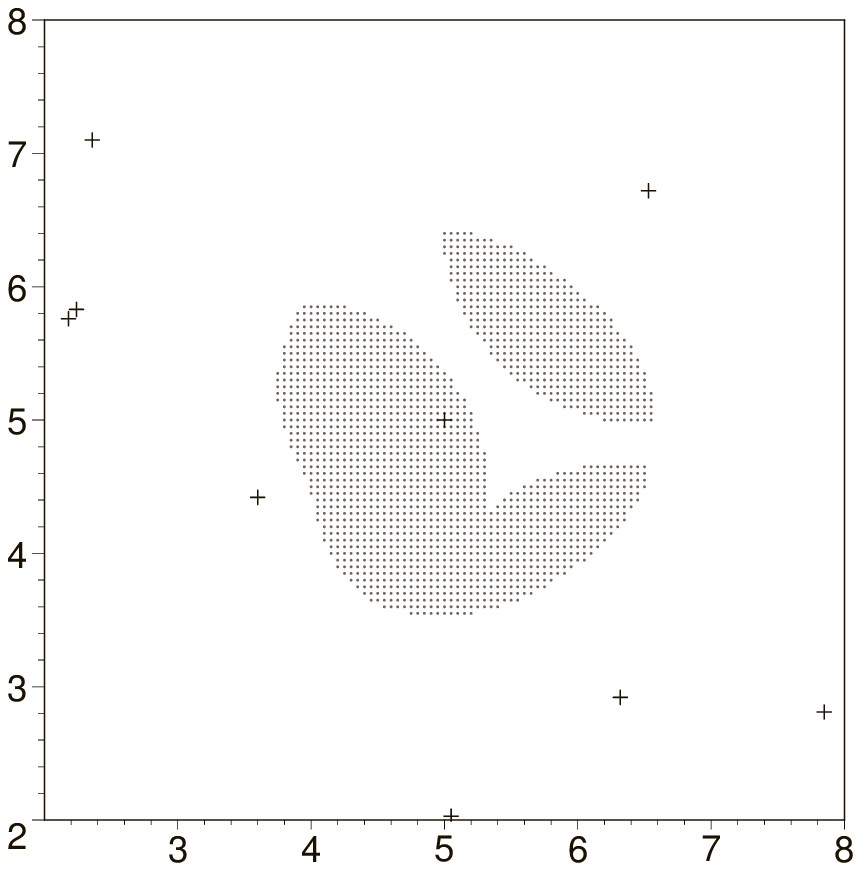}\\
\vspace{.2cm}
(i) \hspace{5cm} (ii)\\
\vspace{.2cm}
\includegraphics[width=0.3\textwidth]{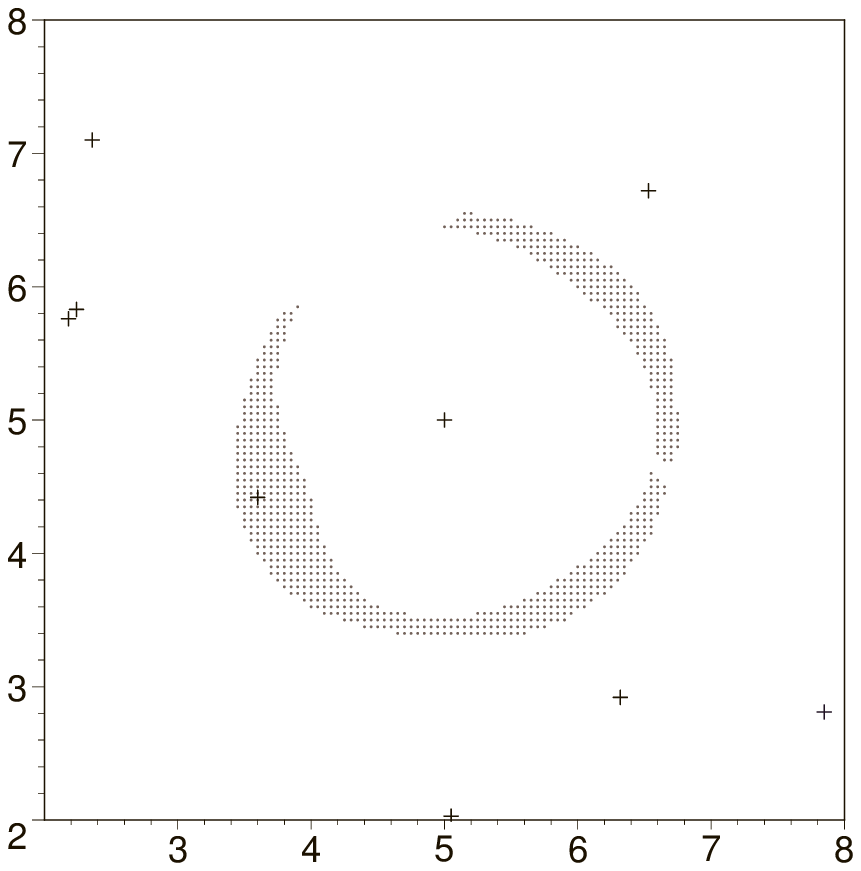}
\includegraphics[width=0.3\textwidth]{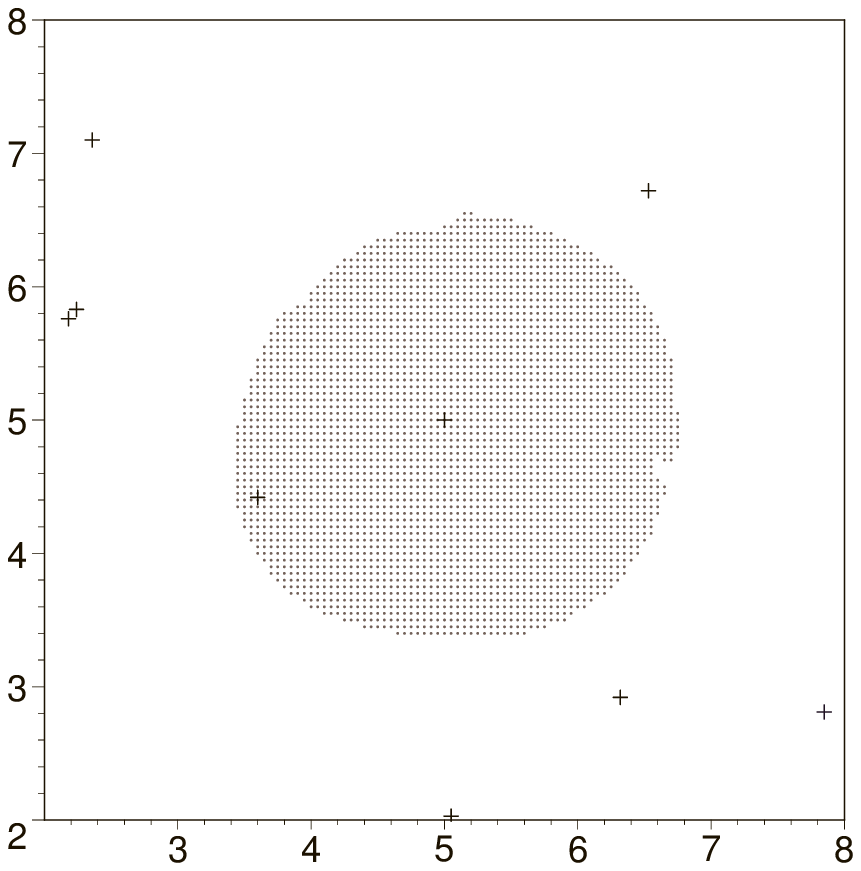}\\
\vspace{.2cm}
(iii) \hspace{4.7cm} (iv)
\end{center}
\caption{
Figure (i) depicts $\Xi^{\JD(1)}$ for the tagged transmitter, located at $[5,5]$, and
Figure (ii) $\Xi^{\JD(1)}\cap \Xi^{\IAN}$.
Figure (iii) shows $\Xi^{\JD(1)\setminus \IAN}$ and
Figure (iv) $\Xi^{\OPT(1)}$. The path loss exponent is $\beta=2.5$, the
power constraint is $Q=100$ for all users; the threshold is $R = 0.2$ bits/s/Hz, and the user density
is $\lambda=0.3$. The attenuation is $l(r)=(1+r)^{-\beta}$.}
\label{fig:sim2}
\end{figure*}

\subsection{SIC($1$) versus OPT($1$)}
There are interesting differences between SIC($1$) and OPT($1$). Let
\[
\Theta^{A}= \bigcup_j \ \Xi_j^A,
\]
where $\Xi_j^A$ is the cell of transmitter $j$ using decoding rule $A$ (IAN, OPT$(1)$, SIC$(1)$).
Also let
\[
\nu^{A}(y) =\sum_j 1_{y \in  \Xi_j^A}
\]
denote the number of transmitters covering location $y$ under condition $A$. Consider the following observations.
\begin{enumerate}
\item We have
\begin{equation}
\Theta^{\SIC(1)}=\Theta^{\IAN},
\end{equation}
that is, the region of in the plane
 covered when treating interference as noise
is \emph{identical} to that of successive interference cancellation. This
follows from the condition for SIC(1), which implies that the location under consideration
is included in the cell of another transmitter in the symmetrical rate case. The gain of SIC(1)
is hence only in the {\em diversity} of transmitters that can be received at any location $y$, i.e.,

\begin{equation}\nu^{\SIC(1)}(y)\ge \nu^{\IAN}(y) \text{ for every $y$}.
\end{equation}
\item We have
\begin{equation}
\Theta^{\OPT(1)}  \supset \Theta^{\IAN}.
\end{equation}
As we see in Figure~\ref{fig:sim1}, this inclusion is strict for some parameter values, that is,
 {\em optimal decoding increases global coverage, whereas successive
interference cancellation does not}.
We also have
\begin{equation}
\nu^{\OPT(1)}(y)\ge \nu^{\IAN}(y).
\end{equation}
\item Finally, we have
\begin{equation}
\Theta^{\OPT(1)} \supset \Theta^{\SIC(1)}.
\end{equation}
There is no general comparison between $\nu^{\SIC(1)}(y)$ and $\nu^{\OPT(1)}(y)$, however.
\end{enumerate}
\subsection{The OPT($k$) Cell}
We now explore the performance of OPT($k$), that is, when the tagged receiver jointly decodes up to the strongest $k$ interferers and treats the rest as noise. In Figure~\ref{fig:sim3}, we give samples of the region
$\Xi^{\OPT(2)\setminus\OPT(1)}$, which is
the additional area covered by moving from OPT($1$) to OPT($2$).

\begin{figure*}[tbp]
\begin{center}
\includegraphics[width=0.3\textwidth,angle=-90]{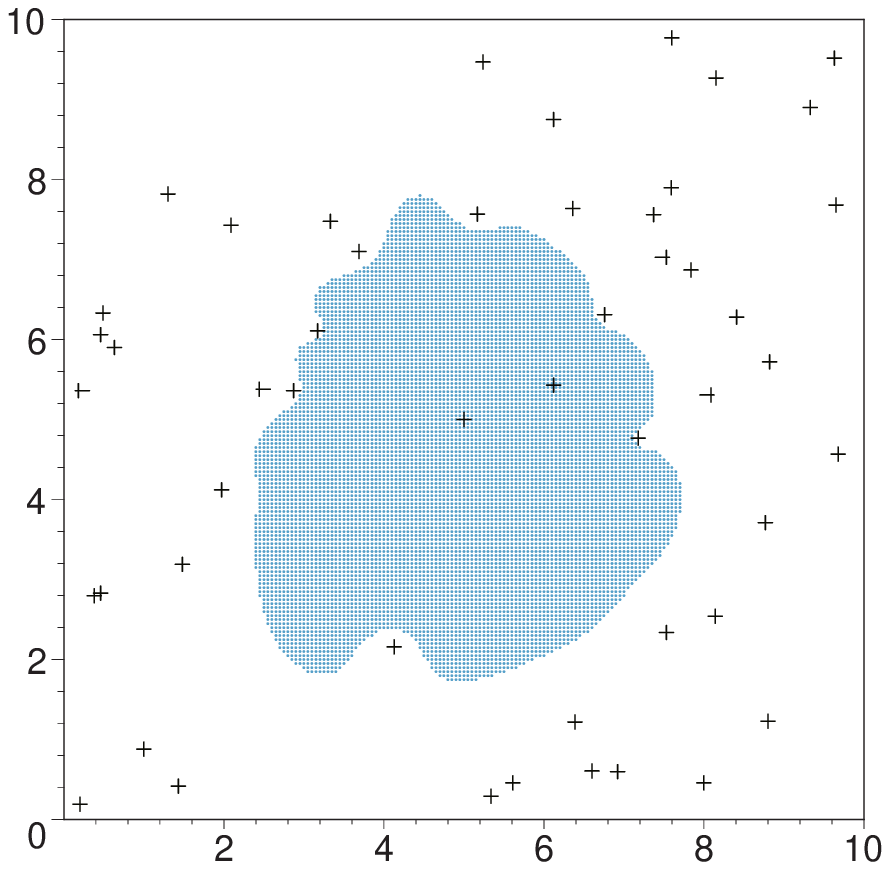}
\includegraphics[width=0.3\textwidth,angle=-90]{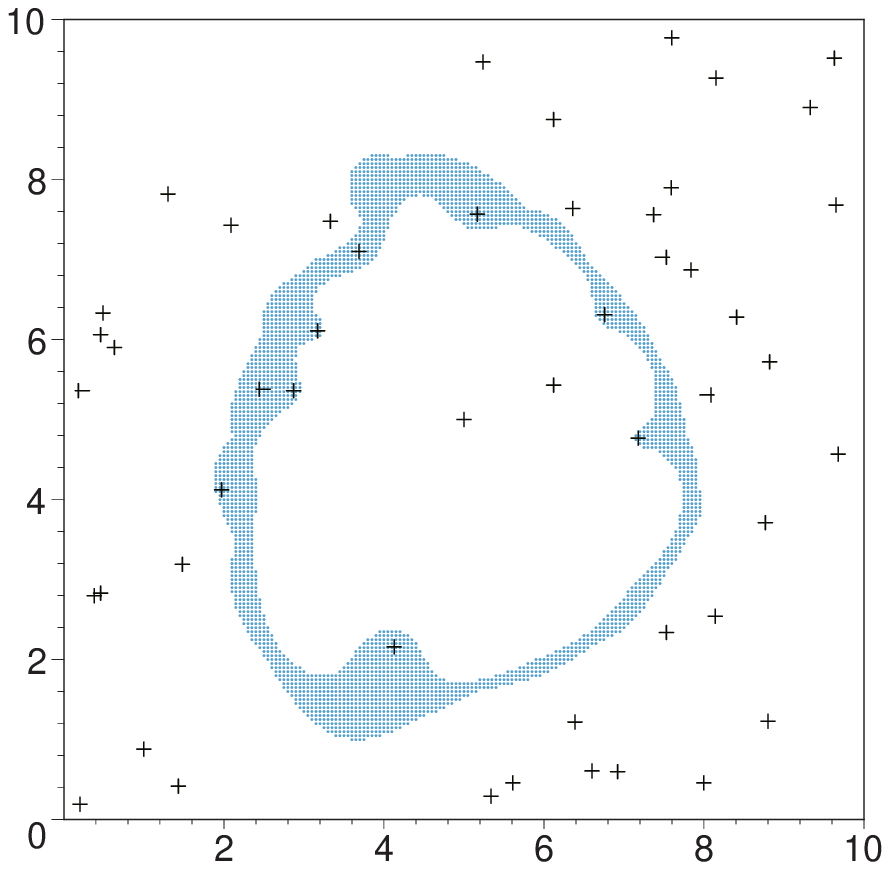}
\includegraphics[width=0.3\textwidth,angle=-90]{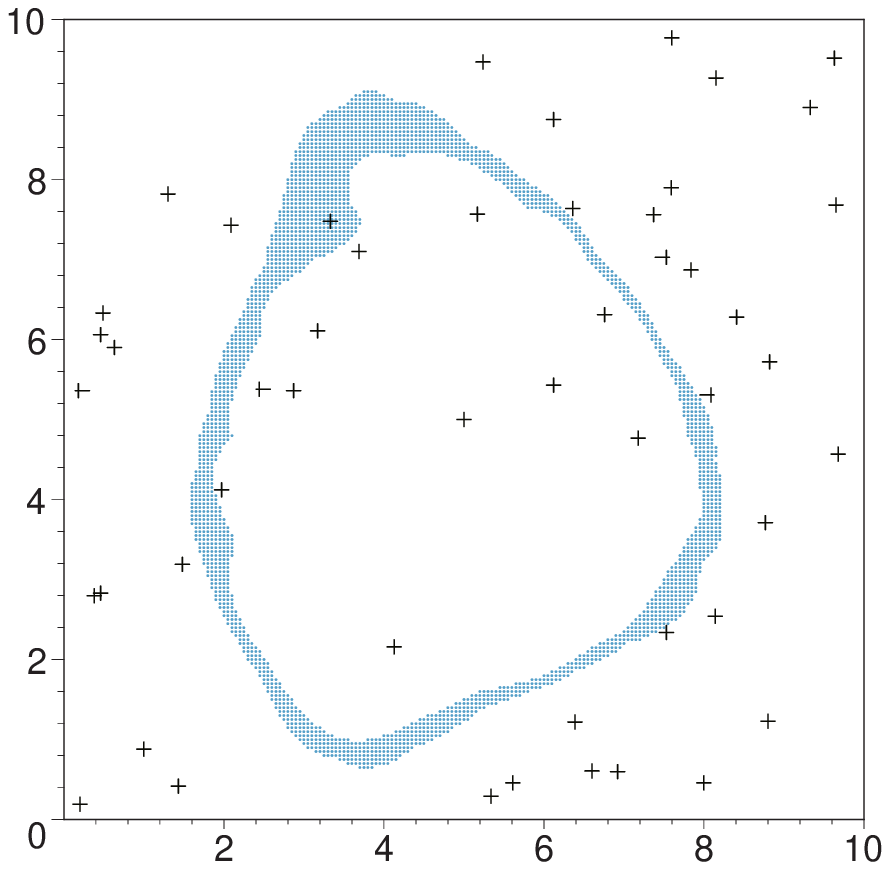}\\
\vspace{.2cm}
(i) \hspace{4.85cm} (ii) \hspace{4.85cm} (iii)\\
\vspace{.2cm}
\includegraphics[width=0.306\textwidth]{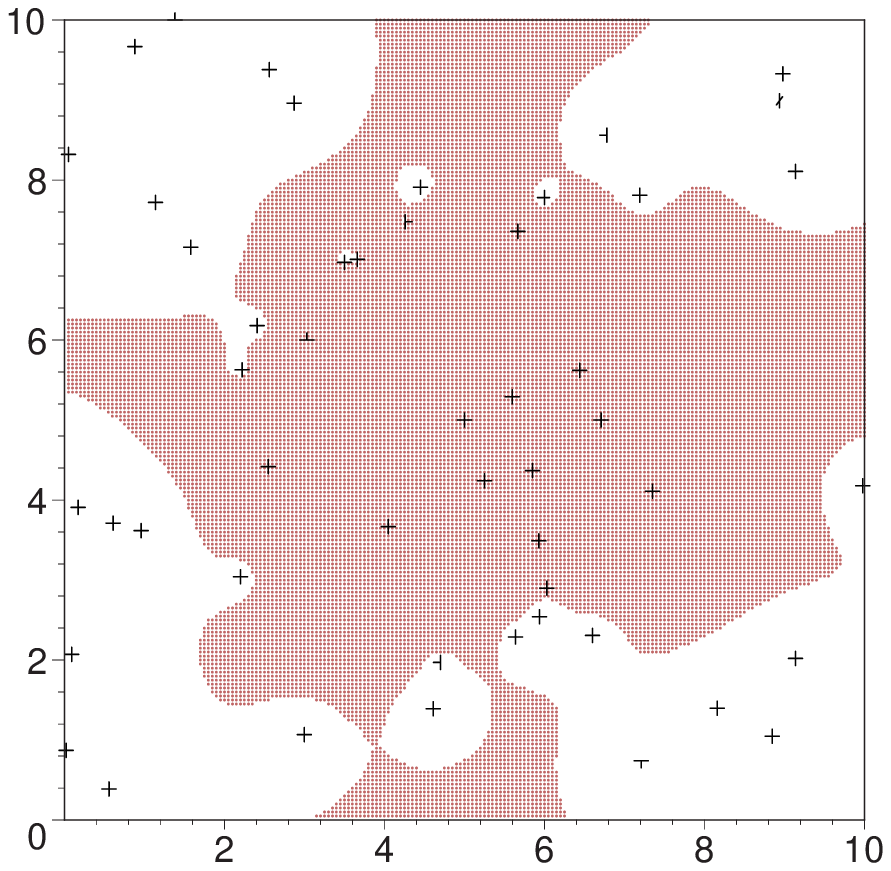}
\includegraphics[width=0.306\textwidth]{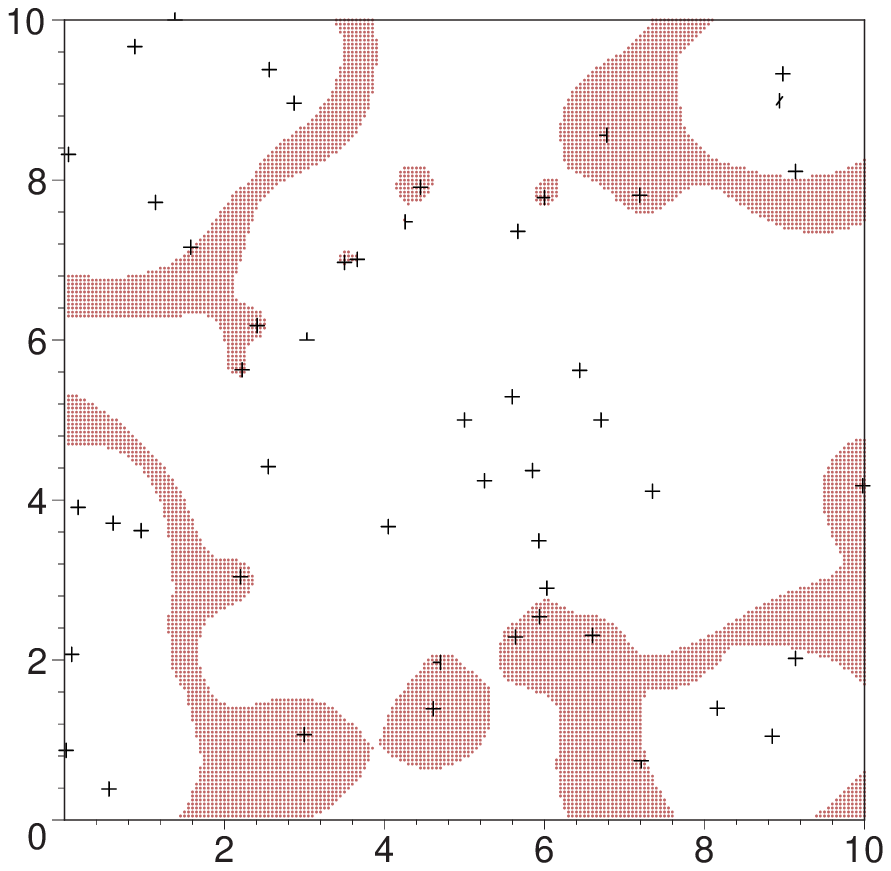}
\includegraphics[width=0.306\textwidth]{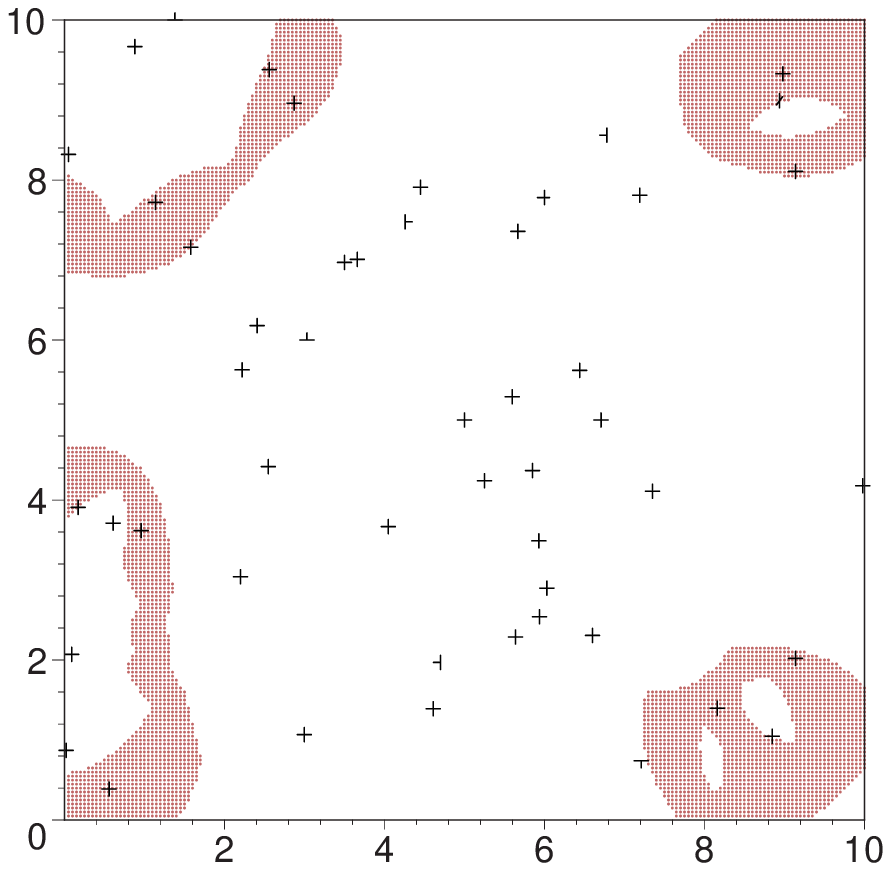}\\
\vspace{.2cm}
(iv) \hspace{4.85cm} (v) \hspace{4.85cm} (vi)
\end{center}
\caption{On (i) and (iv),
the dashed region is $\Xi^{\IAN}$ for the tagged transmitter.
The dashed regions of (ii) and (v) give $\Xi^{\OPT(1)\setminus \IAN}$ for the tagged transmitter;
those of (iii) and (vi) give
$\Xi^{\OPT(2)\setminus \OPT(1)}$ for the tagged transmitter.
The spatial user density is 0.5 and the power constraints
are constant and all equal to $Q= 100$. Here $\beta=3$. The tagged transmitter
is at the center of the plot. The attenuation is $l(r)=(1+r)^{-\beta}$.
The top plots are for $R = 0.03$ bits/s/Hz  and the bottom ones are for $R =0.015$ bits/s/Hz.}
\label{fig:sim3}
\end{figure*}

Since $\Xi^{\OPT(k)}\subset \Xi^{\OPT(k+1)}$ for all $k$,
there exists a limit set $\Xi^{\OPT(\infty)}$, which is the set of locations
where the tagged receiver can decode the message of the tagged transmitter jointly
with some set of other interferers messages at rate $R$
(the existence follows from monotonicity and a boundedness argument using the
assumption that the noise power is positive).
\subsection{Optimal Number of Interferer Messages Decoded}\label{sec:optsim}
We now illustrate the optimal number of jointly decoded interferer messages $k_{\mathrm {opt}}$ defined in
(\ref{eq:optimalk}). Consider transmitters distributed according
a spatial Poisson process with intensity $\lambda=10$.
Assume there is no fading. In {\em {Scenario I}},
the attenuation function is that with a pole and we assume that the path loss exponent is $\beta=3$.
Each transmitter has a power constraint of $Q= 100$ and $P_0=5$ (this means that the distance between the
tagged transmitter and its receiver is appr. 2.71.).
Figure \ref{fiex1-1} plots the function
$k \to \xi(k)$ defined in (\ref{eq:xifunction}) for a sample of a Poisson point process
of interferers. The maximum is reached for
$k_{\mathrm {opt}}=230$ (we just plot the informative part of the curve here)
and $R_\sym$ is approximately 0.00595. The number of interferers with power larger than $P_0$
is $229$. {\em {Scenario II}} is the same but for the attenuation function without pole.
In this case, the number of transmitters with a power larger than $P_0$
is $90$ and $k_{\mathrm {opt}}=90$.
\begin{figure}[h]
\begin{center}
\includegraphics[width=0.24\textwidth]{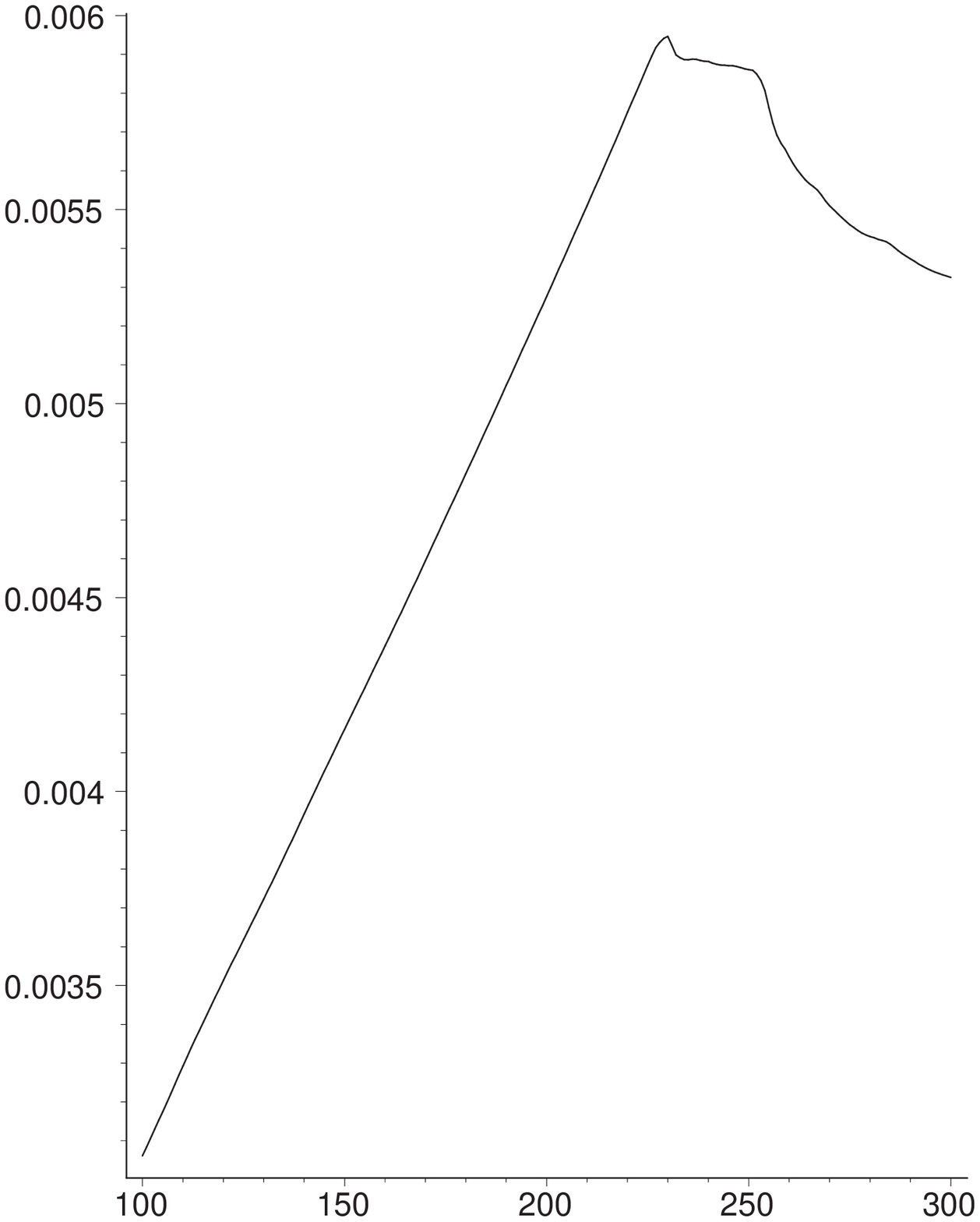}
\includegraphics[width=0.24\textwidth]{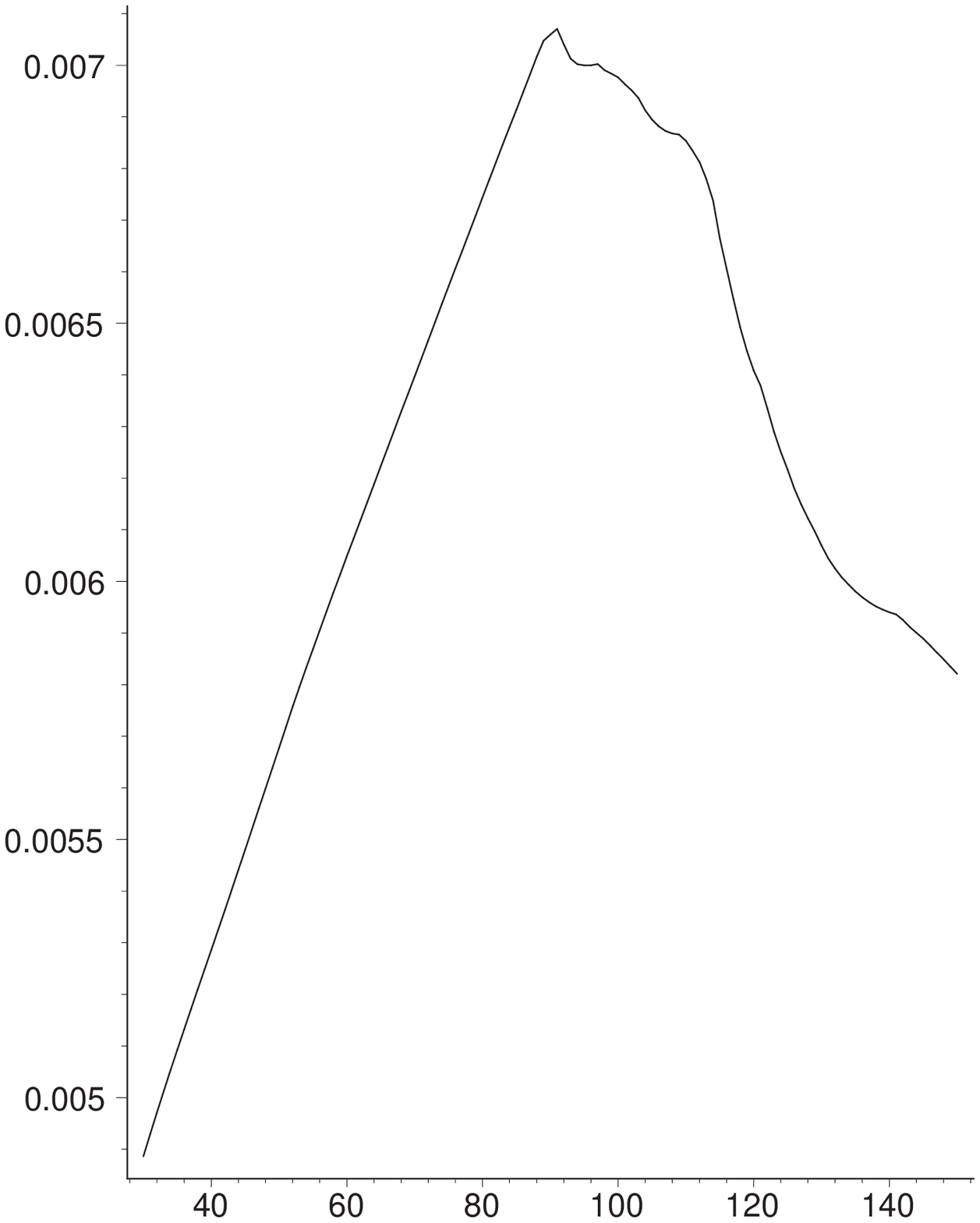}
\end{center}
\caption{A sample of the function $k\to\xi(k)$ for a Poisson collection of interferers.
The $x$-axis is that of the $k$ variable.
The maximum of this function provides $k_{\mathrm {opt}}$. Left: the attenuation is $l(r)=(1+r)^{-\beta}$.
Right: $l(r)=r^{-\beta}$.
}
\label{fiex1-1}
\end{figure}

\subsection{Single Hop in Ad Hoc Networks}\label{sec:scahn}

In order to further illustrate the differences between IAN and OPT,
we consider a simulation setting extending that considered above.
We assume a tagged transmitter and its receiver and
a collection of other transmitters located according to a Poisson point
process of intensity $\lambda$ that represent the nodes of an ad hoc network
that interfere with the tagged transmission.
We wish to compare the largest distance between the tagged transmitter and
its receiver under IAN and OPT.

To do so, we first fix a distance between the tagged transmitter and its receiver,
which determines $P_0$ (as in the last subsection).
We use a Monte Carlo simulation to determine the optimal number of decoded transmitters $k_{\mathrm {opt}}$
for a sample of the Poisson point process of interferers.
We then use (\ref{eq:infinity}) to determine the largest possible achievable
rate $R$ under OPT.

We then consider the largest distance $r$ between the tagged transmitter and receiver
such that the rate $R$ is achievable using IAN, i.e., such that $R< C(Qr^{-\beta}/(1+I))$,
for the interference $I$ created by the same point process of interferers as above.

Figure \ref{fiex1-2} shows the locations of the interferers (obtained
by sampling a Poisson point process), the tagged receiver
(located at the center) and the tagged transmitter (at the
other end of the long segment). The setting is that of Scenario I
of Subsection \ref{sec:optsim}.

The long segment represents the distance between the tagged transmitter to its receiver (tagged link)
under OPT. Its length is approximately 2.71. The short segment
(displayed here for comparison) is the longest possible IAN link
to the same receiver for the same rate $R=0.00595$ of the OPT link.
The length of the latter link is 0.543. Hence, for this setting,
about five times longer single hops can be supported at rate
$R$ when moving from IAN to OPT.

\begin{figure}[ht]
\centering
\includegraphics[width=0.45\textwidth,height=0.45\textwidth]{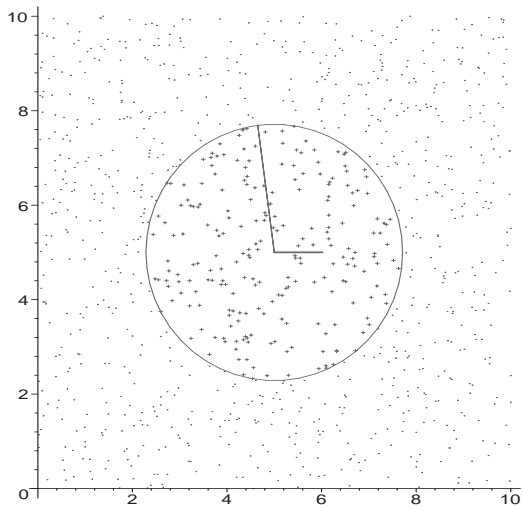}
\caption{Comparison of IAN and OPT in an ad hoc network. The receiver is at the center
of the plot (at [5,5]).
The long segment is a link of fixed length using OPT. The rate $R$ is the largest rate that $R_\sym$
can be sustained on this link.
The short segment gives the longest
link that can be sustained at the same rate $R$ under IAN at this receiver.
The same thermal noise at the receiver, the
transmission power at the transmitters and the set of interferers are the same in the
two cases.
The set of red points (inside the circle) is the optimal set of transmitters that are
jointly decoded by the receiver under OPT, whereas
the set of blue points (outside the circle) features the other transmitters that the tagged receiver
considers as noise. The attenuation is $l(r)=r^{-\beta}$.
}
\label{fiex1-2}
\end{figure}

\section{Asymptotic Analysis in the Wideband Regime}

\newcommand{\bP}{{\bar P}}
\newcommand{\bR}{{\bar R}}
\newcommand{\bQ}{{\bar Q}}

This section is devoted to the analysis of the gain
offered by using the optimal decoder OPT($K$) compared to IAN in large networks using the stochastic
geometry approach \cite{BB09-1}.
The setting is that of Section \ref{secSD}, namely we
consider a tagged transmitter-receiver pair  and a denumerable collection of interferers.
We assume here that these interferers are located according to
some homogeneous Poisson process in the plane.
We focus on the wideband regime, where all users share a large bandwidth and
the density of users is large.

More precisely, the wideband limit is the regime where the bandwidth $B \to \infty$, the average transmit power is fixed at $\bQ$ Watts, and the target data rate for each transmitter is fixed at $\bR$ bits/s. This means that the transmit power per Hz $Q=\bQ/B$ and the data rate per Hz $R= \bR/B$ both tend to zero as $B \to \infty$. We also assume that the tagged transmitter and receiver are at a fixed {\em non-random} distance $r_0$ from each other, so that the {\em received} power from the tagged transmitter at the tagged receiver is fixed at $\bP_0 = l(r_0) \bQ$ Watts. On the other hand, the received power from interferer $j$, $\bI_j = l(T_j-y) \bQ$, is random. The noise power is assumed to be $1$ Watt/Hz, as before.

As the bandwidth $B$ increases, we would like the network to support an increasing user density $\lambda_B$ so that the spectral efficiency of the system is kept non-vanishing. The following two theorems compare the performance of the IAN decoder and the OPT decoder in terms of how fast the density $\lambda_B$ can scale with the bandwidth $B$ while still reliably decoding the tagged transmitter's message.

\begin{theorem}
\label{thm:ian_wb}
Consider the path loss model  $l(r)=r^{-\beta}$. If $\lambda_B=\kappa B^{p}$ with
$p > \frac 2 \beta$ and $\kappa>0$, then
for every target rate  $\bar R>0$, the
IAN decoding condition for the tagged receiver cannot be satisfied almost surely as $B$ grows.
\end{theorem}

\begin{theorem}\label{thesufcfeas}
Consider the path loss model $l(r)=r^{-\beta}$. If $\lambda_B=\kappa B$ for $\kappa>0$ and if
\begin{equation}
\label{eq:g-1}
\rho \bR < C \left(\frac{\rho \bP_0}{1+ 2\rho \bP_0 /(\beta-2)}\right)
= C \left(\frac{\rho \bP_0}{1+
\frac{ 2\pi\kappa \bQ }{r_0^{\beta}(\beta-2)}}\right),
\end{equation}
then almost surely the OPT decoding condition is satisfied as the bandwidth $B$ grows. Here, $\rho = \kappa \pi r_0^2$ is the expected number of interferers per Hz within the communication radius $r_0$ from the tagged receiver.
\end{theorem}

Theorem \ref{thm:ian_wb} says that one needs a {\em sub-linear} scaling of the user density to guarantee a positive rate $\bR$ under the IAN decoder. In particular, the classical strong law of large numbers shows that
for all $\kappa>0$, the linear scaling $\lambda_B =\kappa B$ leads to a zero achievable rate under IAN. This implies that the network spectral efficiency in terms of total bits per second per Hz per unit area goes to zero under IAN. For the OPT decoder, on the other hand, Theorem \ref{thesufcfeas} says that a linear scaling of user density can be supported and hence a positive spectral efficiency can be achieved in the wideband limit.

Before proving the above theorems, we provide some intuition as to why the user density scalings of these two decoders differ quite dramatically. The situation is depicted in Figure \ref{fig:wb_fig}. The received interference power from each of the strong interferers inside the circle is larger than the signal power $P_0$. In fact, as the user density increases, there will be more and more interferers very close to the tagged receiver with much larger received powers than $P_0$. Their effect is fatal to the IAN decoder, which treats all interference as noise. The OPT decoder, on the other hand, can take advantage of these interferers' high received powers to jointly decode their messages together with that of the tagged transmitter. This effectively turns their interference power into useful signal energy. The proof of Theorem \ref{thesufcfeas} shows that the total useful received power from these strong interferers is at least comparable to the total harmful received power of the interferers outside the disk; hence reliable communication at a positive rate $\bR$ bits/s for the tagged receiver (and for everyone else). In fact, the term $\rho \bP_0$ in (\ref{eq:g-1}) is a lower bound on the total power per Hz received from the strong interferers, and the term $2\rho \bP_0 /(\beta-2)$ is the total power per Hz received from the weak interferers outside the disk.

\begin{figure}
\footnotesize
\begin{center}
\psfrag{a}[c]{interferers}
\psfrag{c}[c]{weak}
\psfrag{b}[c]{strong}
\psfrag{r}[b]{receiver 0}
\psfrag{t}[r]{transmitter 0}
\psfrag{r0}[b]{$r_0$}
\includegraphics[scale=0.5]{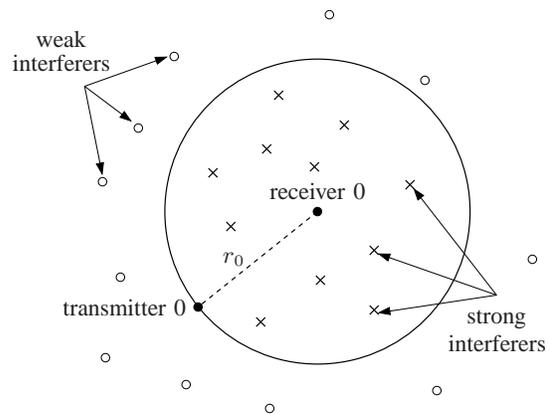}
\end{center}
\caption{The tagged transmitter and the tagged receiver are at a distance $r_0$ from each other. The strong interferers are within a distance of $r_0$ from the tagged receiver.}
\label{fig:wb_fig}
\end{figure}

We are now ready to prove the above theorems.

\begin{myproof}{of Theorem \ref{thm:ian_wb}}
The feasibility of the rate $\bR$ is equivalent to
\begin{equation}
\label{eq:ifffeas}
C\left(\frac{\bP_0} {B + \bar \JJJ}\right) \ge \frac {\bar R}{B},
\end{equation}
where $\bar \JJJ$ denotes the total interference power at the tagged receiver (in Watts)
in a Poisson network of intensity $\lambda_B$.
The last condition is equivalent to
\[
\frac {B \bP_0} {B+ \bar \JJJ}
= \frac {\bP_0} {1+ (\kappa/\lambda_B)^{1/p} \bar \JJJ} \ge B \left[2^{\bar R/B}-1\right].
\]

Let $n=\lfloor \lambda_B \rfloor$. We have $\bar \JJJ  > \sum_{j=1}^n \bar \JJJ_j$ with $\bar \JJJ_j$ a collection of i.i.d. shot noise processes of intensity 1. The random variable $\bar \JJJ_i$ has a stable distribution with parameter $2/\beta$ \cite{Haenggi10}, and hence its moments of order $p$ are infinite for $p>2/\beta$.
The Marcinkievicz-Siegmund strong law of large number \cite{Emb91} then implies that
for all $p>2/\beta$,
\[
\limsup_{n\to \infty}\frac 1 {n^{1/p}} \sum_{j=1}^n \bar \JJJ_j = \infty
\]
in an almost sure sense. Since $ B \left(2^{\bar R/B}-1\right) \to \bR \ln 2 > 0$, this shows that the condition (\ref{eq:ifffeas}) cannot hold true for sufficiently large $B$.
\end{myproof}

\begin{myproof}{of Theorem \ref{thesufcfeas}}
With the notation of Lemma \ref{cor1}, recall that $\lbk$  is the index of the first interferer whose received power at receiver $0$ is less than the received signal power $P_0$. Equivalently, $\lbk-1$ is the number of interferers in a disk of radius $r_0$ from receiver $0$.

From Lemma \ref{cor1}, for the rate $\bR$ to be feasible,
it is enough to show that
\begin{equation}
\label{eq:opt}
\frac{\lbk \bR}{B} < C\left(\frac{\lbk\bP_0 }{B + \bar I[\lbk:\infty]} \right).
\end{equation}
Now, almost surely, when $B$ tends to infinity,
\begin{equation}
\label{eq:kp}
\frac{\lbk}B \to \kappa \pi r_0^2
\end{equation}
and
\begin{equation}
\label{eq:truc}
\lim_{B \to \infty} \frac 1 B \bar I [\lbk:\infty]  =\kappa m(\bP_0) <\infty,
\end{equation}
with
\[
m(\bP_0)=\frac {2\pi r_0^2 \bP_0}{\beta-2}.
\]
In order to show (\ref{eq:truc}), we represent
$\bar \III[\lbk:\infty] $ as the sum of $n$ i.i.d.
random variables $\widetilde \JJJ_1,\ldots,\widetilde \JJJ_n$, where $\widetilde \JJJ_1$ is the
shot noise for the attenuation function $r^\beta$ and for
a spatial Poisson point process with intensity
1 outside a disk of radius $r_0$
and 0 inside. Since $\EE(\widetilde \JJJ_1)<\infty$, (\ref{eq:truc}) follows
from the strong law of large numbers. From Campbell's formula \cite{BB09-1}, we obtain
\[
m(\bP_0)=\EE(\widetilde \JJJ_1)=
\bar Q 2\pi \int_{r_0}^\infty \frac 1 {r^\beta} r dr =
\frac {2\pi r_0^2 \bar P_0}{\beta-2},
\]
using the fact that $\bP_0 = \bQ r_0^{-\beta}$. Substituting (\ref{eq:truc}) and (\ref{eq:kp}) into (\ref{eq:opt}) and simplifying yields the condition (\ref{eq:g-1}) for OPT to decode successfully almost surely for $B$ large enough.
\end{myproof}

Note that the linear user density scaling achieved by OPT cannot be achieved by the decoder OPT($k$) for any fixed $k$. One has to jointly decode the messages from an increasing number of interferers as the bandwidth and the user density  increase.

When the distance $r_0$ between the
tagged transmitter and its receiver tends to infinity, the received power $P_0\rightarrow 0$,
and (\ref{eq:g-1}) reads $\bR<\bP_0 \log e$, which is the
wideband capacity of a point-to-point Gaussian channel without interference.

\noindent {\em Remark:} This result may seem surprising at first glance. In the ad hoc network setting
of Section \ref{sec:scahn}, this result suggests that when OPT($k$) of
high order $k$ is used in a wideband system, one can maintain a channel from
a tagged transmitter to a tagged receiver, say at distance $r_0$,
with a positive rate (determined by Theorem \ref{thesufcfeas}) when
the user density tends to infinity. For instance, in the ad hoc setting
of the simulation section, this means that one can maintain simultaneous single hop
channels that "jump" over a very large number of nodes  of the
ad hoc network. In contrast, in the IAN case, in order to maintain the same rate,
one has to set a multihop route over a number
of relay nodes that tends to infinity as the user density tends
to infinity. However, this is perhaps not so surprising since in
this setting, one could in principle organize some sort
of TDMA or FDMA  IAN scheme (which silences a large collection of nodes when
the tagged transmission takes place) that has asymptotic performance
of the same kind as that exhibited by OPT. So, OPT can in fact be seen as a
way of obtaining good performance without a priori partitioning the users into different time or frequency slots.

Notice that the last comparison results rely on the assumption that
the loss function is the one with a pole.
In the case without a pole, we can obtain the following results using very similar arguments
based on the classical strong law of large numbers (and are easily extended to general
attenuation functions such that $\int_{\RR^+} l(r) r dr <\infty$).

\begin{theorem}\label{theorem:sanspole}
Consider the path loss model $l(r)=(k+r)^{-\beta}$. When $\lambda_B=\kappa B$,
and the bandwidth $B$ tends to infinity, the IAN decoding condition is satisfied almost surely iff
\begin{align}
\label{eq:g-2}
\bar R & <
C\left(\frac {\bP_0} {1 + 2 \rho \bP_0\int_0^\infty \left(\frac{k+r_0}{k+r}\right)^\beta \frac{r}{r_0^2} dr}\right)\nonumber\\
& =
C\left(\frac {\bP_0} {1 + 2 \pi \kappa \bQ \int_0^\infty \left(\frac{1}{k+r}\right)^\beta r dr}\right),
\end{align}
and the OPT condition almost surely if
\begin{align}
\label{eq:g-3}
\rho \bar R & <
C\left(\frac{\rho\bP_0}{1 + 2 \rho \bP_0\int_{r_0}^\infty \left(\frac{k+r_0}{k+r}\right)^\beta \frac{r}{r_0^2} dr}\right)\nonumber \\
& =
C\left(\frac{\rho\bP_0}{1 + 2 \pi \bQ \int_{r_0}^\infty \left(\frac{1}{k+r}\right)^\beta r dr}\right),
\end{align}
where again $\rho = \kappa\pi r_0^2$.
\end{theorem}
As a direct corollary of the last formulas, when the distance $r_0$ between the
tagged transmitter and its receiver tends to infinity so that the received power $P_0\to 0$, the right hand side of (\ref{eq:g-3}) tends to zero like $\bP_0 \log e$, which is the wideband capacity of a point-to-point Gaussian channel without interference. On the other hand, the effect of interference never disappears for IAN. Hence, in this limiting regime, for the case without a pole and with a given linear user growth rate, a positive rate is feasible for both IAN and OPT in the limit, but with different values. When the tagged transmitter and receiver are far away from each other, the scaling of this feasible rate under OPT is as though there were no interferers.

\noindent{\em Remark:}
Above, we focused on the case where the node density tends to infinity. For the finite density case, the performance of the decoding strategies considered can be evaluated from the joint distributions of the total interference and of the order statistics $I_1,I_2,\ldots$ using the tools described in \cite{BB09-1}.

\section{Conclusion}
In this paper, we studied the optimal performance achievable in a Gaussian interference network when the transmitters are constrained to use uncoordinated point-to-point codes. While recent results have shown that to achieve the {\em ultimate} capacity of such networks, techniques such as superposition coding and interference alignment are needed, such techniques require significantly more complex codes and coordination between the transmitters. What our results suggest is that using simple point-to-point codes and no coordination between the transmitters, one can achieve quite significant gains. Moreover, since many existing wireless networks already use near-capacity-achieving point-to-point coding, our results also point to the possibility of significant performance gain from just upgrading the receivers and not the transmitters. This provides an evolutionary path to improving the performance of existing wireless networks. It would also be interesting to extend our results to establish the capacity region for MAC-capacity-achieving codes with limited coordination, such as frequency/time partitioning and power control.

An interesting future direction is to explore how to design a distributed medium-access protocol when receivers employ optimal decoding. There would be two important components to such a protocol. The first component is {\em interferer sensing} by the receivers. Each receiver senses the powers and the identities of its interferers. This can be implemented through some beaconing scheme. The second component is a  {\em backoff} procedure by the transmitter. Each user, when it has data to transmit, needs to sense when its receiver can  accommodate its transmission. This in turn depends on the number and powers of the interferers who are transmitting. In conventional protocols such as Carrier Sense Multiple Access (CSMA), transmission occurs when the level of interference is below a certain threshold. This makes sense for an IAN receiver. However, under optimal decoding, sometimes having a strong interferer is advantageous as it enables joint decoding. Hence, the backoff procedure will have to be more elaborate.

Another interesting direction is to explore the implementation of optimal decoding. When the number of interferers whose messages are jointly decoded becomes large, one might fear an exponential growth of the combination of codewords to be tested by each decoder when decoding. However, it is not completely clear that this exponential growth is necessary to achieve capacity. For example, at the corner points, SIC, with a complexity that only grows linearly with the number of decoded interferers, is sufficient. It may be possible to reduce the complexity of decoding for the points in the interior of the sum rate face as well.

\bibliographystyle{plain}

\end{document}